\documentclass{article}

\usepackage{amsmath, amsthm, amssymb}
\usepackage{algorithm}
\usepackage[noend]{algorithmic}
\newtheorem{theorem}{Theorem}

\newtheorem{lemma}{Lemma}
\theoremstyle{definition}
\newtheorem{definition}{Definition}

\usepackage{arxiv}

\usepackage[utf8]{inputenc} 
\usepackage[T1]{fontenc}    
\usepackage{hyperref}       
\usepackage{url}            
\usepackage{booktabs}       
\usepackage{amsfonts}       
\usepackage{nicefrac}       
\usepackage{microtype}      
\usepackage{cleveref}       
\usepackage{lipsum}         
\usepackage{graphicx}
\usepackage{natbib}
\usepackage{doi}

\title{Scalable Edge Blocking Algorithms for Defending Active Directory Style Attack Graphs}
\author{Mingyu Guo, Max Ward, Aneta Neumann, Frank Neumann, Hung Nguyen
\\
\\
    1. School of Computer Science, University of Adelaide, Australia\\
    \{mingyu.guo, aneta.neumann, frank.neumann, hung.nguyen\}@adelaide.edu.au\\
    \\
    2. School of Physics, Maths and Computing, Computer Science and Software Engineering\\
    University of Western Australia, Australia\\
    {\em and}\\
    Department of Molecular and Cellular Biology\\
Harvard University, Cambridge, Massachusetts, USA\\
    max.ward@uwa.edu.au
}

\begin{document}

\maketitle

\begin{abstract}
Active Directory (AD) is the default security management system for Windows
    domain networks.  An AD environment naturally describes an
    attack graph where nodes represent computers/accounts/security groups, and edges represent existing accesses/known exploits that allow the attacker
    to gain access from one node to another.  Motivated by practical AD use
    cases, we study a Stackelberg game between one attacker and one defender.
    There are multiple entry nodes for the attacker to choose from and there is
    a single target (Domain Admin).  Every edge has a failure rate.  The
    attacker chooses the attack path with the maximum success rate.  The
    defender can block a limited number of edges (i.e., revoke accesses) from a set of blockable edges, limited by budget. The defender's aim is to minimize the
    attacker's success rate.

We exploit the
    tree-likeness of practical AD graphs to design scalable algorithms.
We propose two novel methods that combine theoretical fixed parameter analysis
    and practical optimisation techniques.

For graphs with small tree widths, we propose a tree decomposition based
    dynamic program.  We then propose a general method for converting tree
    decomposition based dynamic programs to reinforcement learning
    environments, which
    leads to an anytime algorithm that scales better, but loses the optimality
    guarantee.

For graphs with small numbers of non-splitting paths (a parameter we invent
    specifically for AD graphs), we propose a kernelization technique that
    significantly downsizes the model, which is then solved via
    mixed-integer programming.

Experimentally, our algorithms scale to handle synthetic AD graphs
    with tens of thousands of nodes.

\end{abstract}


\section{Introduction}

Active Directory (AD) is the {\em default} system for managing access
and security in Windows domain networks.
Given its prevalence among large and small organisations
worldwide, Active Directory has become a major target by cyber attackers.\footnote{Enterprise Management Associates~\cite{EnterpriseManagementAssociatesEMA2021:Rise} found that 50\% of organizations surveyed had experienced an AD attack since 2019.}
An AD environment naturally describes a {\em cyber attack graph}
--- a conceptual model for describing the causal relationship of cyber events.
In an AD graph, the nodes are computers/user
accounts/security groups. An edge from node A to node B represents that an
attacker can gain access from A to B via an existing access or a known exploit.
Unlike many attack graph models that are of theoretical interest only,\footnote{\cite{Lallie2020:reviewa} surveyed over
$180$ attack graphs/trees from academic literatures on cyber security.}
AD attack graphs are actively
being used by real attackers and IT admins. Several software tools
(including both open source and commercial software)
have been
developed for scanning, visualizing and analyzing AD graphs. Among
them, one prominent tool is called {\sc
BloodHound},
which models the {\em identity snowball
attack}. In such an attack, the attacker starts from a low-privilege account,
which is called the attacker's entry node (i.e., obtained via {\em phishing}
emails).  The attacker then travels from one node to another, where the end
goal is to reach the highest-privilege account called the {\sc Domain Admin
(DA)}.

Given an entry node, {\sc BloodHound}
generates a {\em shortest attack path} to {\sc DA}, where
a path's distance is equal to the number of hops (fewer hops implies less chance of failure/being detected).  Before the invention of
{\sc BloodHound}, attackers used personal experience and heuristics to explore the attack graph,
hoping to reach {\sc DA} by chance.
{\sc BloodHound} makes it easier to attack Active Directory.

Besides the attackers, defenders also study Active Directory attack graphs.
The original paper that motivated {\sc BloodHound}~\cite{Dunagan2009:Heatray}
proposed a heuristic for blocking edges of Active Directory attack graphs. The
goal is to cut the attack graph into multiple disconnected regions, which
would prevent the attacker from reaching {\sc DA}.
In an
AD environment, edge blocking is achieved by revoking accesses or
introducing monitoring.
{\em Not all edges are blockable}.
Some accesses are required for the organisation's normal operations.
{\em Blocking is also costly} (i.e., auditing may be needed before blocking an edge).

We study how to {\em optimally} block a limited number of edges in order to
minimize a strategic attacker's success rate (chance of success) for reaching {\sc DA}.  In our
model, we assume that different edges have different {\em failure
rates}.
The defender can block a {\em blockable} edge to increase the
failure rate of that edge (from its original failure rate) to $100\%$. The defender can block at most $b$
edges, where $b$ is the defensive budget. For example, if it takes 1 hour
for an IT admin to block an edge (auditing, reporting, implementation)
and one eight-hour day is dedicated to AD cleanup, then $b=8$.

We study both {\em pure} and {\em
mixed} strategy blocking. A pure strategy blocks $b$ edges
deterministically.  A mixed strategy specifies multiple sets of $b$
edges, and a distribution over the sets.  We follow the standard
Stackelberg game model by assuming that the attacker can observe the defender's
strategy and play a best response. For mixed strategy blocking, the attacker
can only observe the probabilities, not the actual realizations.
There is a set of entry
nodes for the attacker to choose from.
The attacker's strategy specifies an entry node and from it an attack path to {\sc DA}.
The attacker's goal is to maximize the success rate by choosing the best path.

The pure strategy version of our model can be reduced to the {\em
single-source single-destination shortest path edge interdiction} problem,
which is known to be
NP-hard~\cite{Bar-noy1995:Complexity}.  However, NP-hardness on general graphs
does not rule out efficient algorithms for practical AD
graphs.  {\em Active Directory style attack graphs exhibit special graph structures
and we can exploit these structures to derive scalable algorithms.}
We adopt {\em fixed-parameter analysis}.
Formally, given an NP-hard problem with problem size $n$,
let the easy-to-solve instances be characterized by special
 parameters $k_1,\ldots,k_c$. If we are able to derive an algorithm that
 solves these instances in $O(f(k_1,\ldots,k_c)\textsc{poly}(n))$, then we
 claim that the problem is {\em fixed-parameter tractable (FPT)} with respect to
 the $k_i$s.  Here, $f$ is an arbitrary function that is allowed to be
 exponential.  We do require that the running time be polynomial in
 the input size $n$. That is, an FPT problem is practically solvable for large
 inputs, as long as the special parameters are small (i.e., they indeed describe
 {\em easy} instances).


{\em It should be noted that
this paper's focus is not to push the frontier of theoretical fixed-parameter analysis.
Instead, fixed-parameter analysis is our {\bf means} to design scalable algorithms for our specific application on AD graphs.}
As a matter of fact, our approaches combine theoretical fixed-parameter analysis and practical optimisation techniques.

We observe that practical AD graphs have two noticeable
structural features.\footnote{We have included an example synthetic AD graph in the appendix, which is generated using {\sc BloodHound}'s AD graph generator {\sc
DBCreator}.
The AD environment of an organisation is
considered sensitive, so in this paper, we only reference
AD graphs generated using {\sc DBCreator} and {\sc adsimulator}, which is another open source tool for generating synthetic AD graphs.}
We first observe that {\em the attack paths tend to be short}.
Note that we are not claiming that long paths do not exist (i.e., there are cycles in AD graphs).
The phrase ``attack paths'' refer to shortest paths that the attacker would actually use.
The {\sc BloodHound} team uses the phrase ``six degrees
of domain admin'' to draw the analogy to the famous ``six degree of
separation'' idea from the small-world problem~\cite{Milgram1967:SmallWorld} (i.e., all people in this world are on average six or fewer
social connections away from each other).  That is, in an organisation, it is expected
that it takes only a few hops to travel from an intern's account to the CEO's
account. Similar to the ``small-world'' hypothesis, attack paths being short is an {\em unproven observation} that we expect to hold true for practical purposes.

The second structural feature is that AD graphs are very similar to trees.
The tree-like structure comes from the fact that it
is considered a best practise for the AD environment to follow
the organisation chart. For example, human resources
would form one tree branch while marketing would form another tree branch.
However, an Active Directory attack graph is almost never exactly a tree,
because there could be valid reasons for an account in human resources to
access data on a computer that belongs to marketing.  We
could interpret Active Directory attack graphs as {\em trees with extra {\em non-tree edges}
that represent security exceptions}.

Our aim is to design {\em practically scalable} algorithms for
optimal pure and mixed strategy edge blocking. For organisations with
thousands of computers in internal networks, the AD graphs
generally involve {\em tens of thousands of nodes}.
We manage to scale to such magnitude by exploiting the
aforementioned structural features of practical AD graphs.

We first show that having short attack paths alone is not enough to derive
efficient algorithms.  Even if the maximum attack path length is a constant,
both pure and mixed strategy blocking are NP-hard.  We then focus on exploring
the tree-likeness of practical AD graphs.

Our first approach focuses on pure
strategy blocking only.
For graphs with small tree widths, we propose a tree decomposition based
dynamic program, which scales better than existing algorithms from \cite{Guo22:Practical}.  We then propose a general method for converting tree
decomposition based dynamic programs to reinforcement learning environments.
When tree widths are small, the derived reinforcement learning environments'
observation and action spaces are both small. This leads to an anytime
algorithm that scales better than dynamic program, but loses the optimality
guarantee.

Our second approach handles both pure and mixed strategy blocking.
We invent a non-standard fixed parameter
specifically for our application on AD graphs.  A typical attack path describes a {\em privilege escalation} pathway.
It is rare for a node to have
more than one {\em privilege escalating} out-going edges (as such edges often represent security exceptions or misconfigurations). We observe that practical AD graphs consist of {\em non-splitting paths} (paths where
every node has one out-going edge).  For graphs with small numbers of
non-splitting paths, we propose a kernelization technique that significantly
downsizes the model, which is then solved via mixed-integer programming.  We experimentally verify that this
approach scales {\em exceptionally well} on synthetic AD graphs generated by two
open source AD graph generators ({\sc DBCreator} and {\sc adsimulator}).

\section{Related Research}


\cite{Guo22:Practical} studied edge
interdiction for AD graphs, where the
attacker is given an entry node by nature.
In
\cite{Guo22:Practical}, the defensive
goal is to maximize the attacker's {\em expected} attack path length (i.e., number of hops). In this paper,
the defensive goal is to minimize the attacker's {\em worst-case} success rate, based on the assumption that different edges may have different failure rates.
The authors proposed a tree decomposition based
dynamic program, but it only applies to acyclic AD graphs.
Practical AD graphs do contain cycles so this algorithm
does not apply. The authors
resorted to graph convolutional neural
network as a heuristic to handle large AD graphs with cycles.
{\em Our proposed algorithms can
handle cycles, scale better, and produce the optimal results (instead of being mere
heuristics).}  Furthermore, \cite{Guo22:Practical}
only studied pure strategy blocking.

\cite{Goel22:Defending} studied a different model on edge interdiction for AD graphs.
Under the authors' model, both the attacker's and the defender's problem are \#P-hard, and the authors proposed a defensive heuristic based on combining neural networks and diversity evolutionary computation.

The model studied in this paper is similar to the {\em bounded length cut}
problem studied in \cite{Golovach2011:Paths} and
\cite{Dvorak2018:Parameterized}, where the goal is to remove some edges so that
the minimum path length between a source and a destination meets a minimum threshold.
\cite{Dvorak2018:Parameterized} proposed a tree decomposition based dynamic program for the bounded length cut problem. The authors' algorithms
require that all source and destination nodes be added to every bag in the
tree decomposition. {\em This is fine for showing the theoretical existence
of FPT algorithms, but it is practically not scalable.}
Furthermore, for bounded length cut, if a path is shorter than
the threshold, then it {\em must be cut} and if a path is longer than the threshold,
then it {\em can be safely ignored}. This is a much clearer picture than our model
where we need to judge the relative importance
of edges and spend the budget on the most vital ones.


\cite{Jain2011:Double} proposed a double-oracle algorithm
for equilibrium calculation on attack graphs, whose model is defined differently.
Their
approach is designed for general graphs so it only scales to a few hundred
nodes and therefore is not suitable for practical AD graphs.
\cite{Aziz2018:Defender,Aziz2017:Weakening} studied
node interdiction for minimizing inverse geodesic length. \cite{Durkota2019:Hardening} and \cite{Milani2020:Harnessing} studied deception based defense on cyber
attack graphs.

\section{Formal Model Description}

We use a directed graph $G=(V,E)$ to describe the Active Directory environment.
Every edge $e$ has a failure rate $f(e)$.
There is one destination node {\sc DA} (Domain Admin).  There are $s$ entry nodes.  The
attacker can start from any entry node and take any route. The attacker's goal
is to maximize the success rate to reach {\sc DA}, by picking an optimal entry node and an optimal attack path.
The defender picks $b$ edges to block from a set of blockable edges
$E_b\subseteq E$, where $b$ is the defensive budget. The aim of the defender is
to minimize the attacker's success rate. 
\begin{figure}[h]
    \centering
  \includegraphics[width=0.5\linewidth]{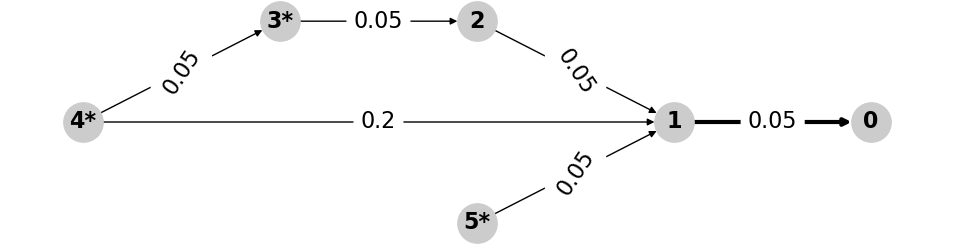}
    \caption{Example attack graph. Node $0$ is {\sc DA}. Node $4,3,5$ are entry nodes (marked using $*$).
    Edge labels represent the edges' failure rates.
    Thick edges (i.e., $1\rightarrow 0$) are not blockable.
    In the appendix, we illustrate how to optimally block $2$ edges from this graph for both pure and mixed strategy.
    }
  \label{fig:example}
\end{figure}

We study both pure and mixed strategy blocking.  We use $B(e)$ to describe the
probability that a blockable edge $e\in E_b$ is blocked.
For pure strategy blocking, $B(e)$ equals either $0$ or $1$.
For mixed strategy blocking, $B(e)$ is in $[0,1]$.
The budget constraint is $\sum_{e\in E_b}B(e)\le b$.
We adopt the standard Stackelberg game model by assuming that the attacker can
observe the defensive strategy and then plays a best response.
For a mixed strategy defense, we assume that the attacker can observe $B(e)$'s probabilistic values,
but not the realisations.

Given $B$, the attacker's optimal attack path can be found via
$\max_{p\in P}\left\{\prod_{e \in p}(1-f(e))(1-B(e))\right\}$,
where $P$ is the
set of all attack paths from all entry nodes.
This maximization problem is equivalent to
$\min_{p\in P}\left\{\sum_{e \in p}\left(-\ln(1-f(e))-\ln(1-B(e))\right)\right\}$.
By applying natural log to convert from product to sum,
we treat an edge's ``distance'' as $-\ln(1-f(e))-\ln(1-B(e))$ (nonnegative).
The attacker's optimal attack path can be solved using Dijkstra's shortest path algorithm~\cite{Dijkstra1959:note}.
Let $\textsc{SR}(B)$ be the success rate of the attacker facing blocking policy $B$.
The defender's problem is $\min_{B}\textsc{SR}(B)$.







Earlier we mentioned that our plan is to exploit the
special structural features of practical AD graphs.
Our first result
is a negative result, which shows that having short attack paths alone is not
enough to derive efficient algorithms. That is, we do need
to consider the tree-like features.

\begin{theorem}
Both pure and mixed strategy blocking are NP-hard for constant maximum attack path length.
\label{thm:nphard}
\end{theorem}

Proof is deferred to the appendix due to space constraint.

\section{Tree Decomposition based Dynamic Program for Pure Strategy Blocking}

{\em Tree decomposition} is a process that
converts a general graph to a tree, where every tree node is a bag (set) of
graph vertices.
The maximum bag size minus one is called the {\em tree width}.
A small tree width indicates that the graph is close to a tree.  Many NP-hard
combinatorial problems for general graphs become tractable if we focus
on graphs with small tree widths. We show that this is also true for our
model.  For the rest of this section, we assume the readers are already familiar
with tree decomposition related terminologies.
We have included all relevant background information regarding
tree decomposition in the appendix, including all algorithms/definitions/terminologies referenced in this paper ({\em i.e.}, {\bf nice} tree decomposition, {\bf introduce}/{\bf forget}/{\bf join} node).
We have also included a running example of our algorithm, pseudocode and relevant proofs in the appendix.

Throughout the discussion, we use {\em nodes} to refer to tree nodes in the tree
decomposition and {\em vertices} to refer to vertices in AD graphs.

Besides assuming a small tree width, another key assumption of our dynamic
program is that we assume a path's success rate is from a small set of at most
$H$ values. A path's success rate is the attacker's success rate for
going through it without any blocking.
If all edges have the same failure rate, then
$H$ is just the {\em maximum attack path length} $l$ plus $2$ ($0$ to $l$ hops, plus
``no path'').
In general, if the number of edge types is a small constant $k$, then $H\in O(l^k)$.
In our experiments, we assume that there are two types of edges
(high-failure-rate and low-failure-rate edges), which corresponds to $H\in O(l^2)$.
It should be noted that $H$ is only used for {\em worst-case} complexity analysis.
In experiments, given a specific graph, a path's number
of possible success rates is often significantly less (i.e.,
if the path is not blockable altogether, then there
is only one possible success rate).

We call our DP {\sc TDCycle} (tree decomposition with cycles\footnote{
This is to differentiate from the dynamic program proposed in \cite{Guo22:Practical}, which cannot handle cycles and does
not guarantee correctness for practical AD graphs as they do contain cycles.}).
The first step is to treat the attack graph as an undirected graph
and then
generate a tree decomposition.
It should be noted that the optimal tree
decomposition with the minimum tree width is NP-hard to
compute~\cite{Arnborg1987:Complexity}. In our experiments, we adopt the
{\em vertex elimination} heuristic for generating tree
decomposition~\cite{Bodlaender2006:Exact}.
We then convert the resulting tree
decomposition into a {\em nice tree
decomposition}~\cite{Cygan2015:Parameterized}, where the root node is a bag
containing {\sc DA} only and all the leaf nodes are bags of size one.
Please refer to the appendix for the process details.
We use {\sc TD} to denote the resulting nice tree decomposition.
{\sc TD} has $O(wn)$ nodes where $w$ is the tree width.

\begin{lemma}
    \label{lm:forget1to1}
    Let $(u,v)$ be an arbitrary edge from the original AD graph. Under {\sc TD}, there exists one and only one
    {\em forget node} $X$, whose child is denoted as $X'$, where $\{u,v\}\subseteq X'$ and $X'\setminus X$ is either $\{u\}$ or $\{v\}$.
\end{lemma}

The above lemma basically says that every edge can be ``assigned'' to exactly
one forget node.
For forget node $(x_2,x_3,\ldots,x_k)$ with child $(x_1,x_2,\ldots,x_k)$ (i.e., $x_1$
is forgotten), we assign all edges between $x_1$ and one of $x_2,\ldots,x_k$ to
this forget node.
The high-level process of our dynamic program
is that we first remove all edges from the graph. We then go through {\sc TD}
{\bf bottom up}. At forget node $X$, we examine all the edges assigned to
$X$. If an edge is not blockable or we decide not to block it, then we put it
back into the graph.  Otherwise, we do not put it
back.  After we finish the whole tree (while ensuring that the budget spent
is at most $b$ $\iff$ we have put back at least
$|E|-b$ edges), we end up with a complete blocking policy.

Let $X=(x_1,x_2,\ldots,x_k)$ be a tree node ($k\le w+1$).  Let $St(X)$ be the subtree of
{\sc TD} rooted at $X$.  Let $Ch(X)$ be the set of all graph vertices
referenced in $St(X)$.  Let $Ch(X)'=Ch(X)\setminus X$.  $Ch(X)'$ is then the
set of vertices already forgotten after we process $St(X)$ in the bottom-up fashion.  A
known property of tree decomposition is that the vertices in $Ch(X)'$ cannot
directly reach any vertex in $V\setminus Ch(X)$.
That is, any attack path from
an entry vertex in $Ch(X)'$ to {\sc DA} must pass through some $x_i$ in $X$.
Also, any attack path (not necessarily originating from $Ch(X)'$) may involve
vertices in $Ch(X)'$ by entering the graph region form by $Ch(X)'$ via
 a certain $x_i$ and then exit the region via a different node $x_j$.
An attack path may ``enter and exit'' the region multiple times but
all entries and exists must be via the vertices in $X$.

Suppose
we have spent $b'$ units of budget on (forget nodes of) $St(X)$. We do not
need to keep track of the specifics of which edges have been blocked. We
only need to track the total budget spending and the following ``distance'' matrix:
\begin{footnotesize}
\[
    M=\begin{bmatrix}
        d_{11} & d_{12} & \ldots & d_{1k}\\
        \ldots\\
        d_{k1} & d_{k2} & \ldots & d_{kk}\\
    \end{bmatrix}\]
\end{footnotesize}

$d_{ij}$ represents the minimum path distance\footnote{Recall
that given an edge with failure rate $f(e)$, we treat the edge's ``distance''
as $-\ln(1-f(e))$ when it is not blocked.} between $x_i$ and $x_j$, where the intermediate edges used
are the edges we have already put back after processing $St(X)$.  Diagonal element $d_{ii}$ represents the
minimum path distance from any entry vertex (among $Ch(X)$) to $x_i$.
We say the tuple $(M,b')$ is {\em possible} at $X$ if and
only if it is possible to spent $b'$ ($b'\le b$) on $St(X)$ to achieve the distance matrix
$M$.

Every tree node of {\sc TD} corresponds to a DP subproblem and there
are $O(wn)$ subproblems.  The subproblem corresponding to $X$ is denoted
as $DP(X)$. $DP(X)$ simply returns the collection of all possible tuples at node $X$.

\noindent{\bf Base cases:} For a leaf node $X=\{x\}$,
if $x$ is an entry vertex, then $DP(X)$ contains one tuple, which is $([0],0)$.
Otherwise, the only possible tuple in $DP(X)$ is $([\infty],0)$.

\noindent{\bf Original problem:} The root of {\sc TD} is $\{\textsc{DA}\}$.
The original problem is then $DP(\{\textsc{DA}\})$, which returns the collection
of all possible tuples at the root. Every tuple from the collection has the form $([d_{\textsc{DA,DA}}],b')$,
which represents that it is possible to spend $b'$ to ensure that the attacker's distance
from {\sc DA} is $d_{\textsc{DA,DA}}$. The maximum $d_{\textsc{DA,DA}}$
in $DP(\{\textsc{DA}\})$
corresponds to the attacker's success rate facing optimal blocking.

\vspace{.05in}
\noindent We then present the recursive relationship for our DP:

\noindent{\bf Introduce node:} Let $X=(x_1,\ldots,x_k,y)$ be an introduce node,
whose child is $X'=(x_1,\ldots,x_k)$. Given a possible tuple in $DP(X')$, we generate a new tuple as follows, which should belong to $DP(X)$. $d_{yy}$ is $0$ if $y$
is an entry vertex and it is $\infty$ otherwise (when $y$ is introduced, all its edges have not been put back yet so it is disconnected from the $x_i$).
\begin{footnotesize}
\[
    \left(\begin{bmatrix}
        d_{11} & \ldots & d_{1k}\\
        \ldots\\
        d_{k1} & \ldots & d_{kk}\\
    \end{bmatrix},b'\right)
    \rightarrow
    \left(\begin{bmatrix}
        d_{11} & \ldots & d_{1k} & \infty\\
        \ldots\\
        d_{k1} & \ldots & d_{kk} & \infty\\
        \infty  & \ldots & \infty & d_{yy} \\
    \end{bmatrix},b'\right)
\]
\end{footnotesize}

\noindent{\bf Forget node:}
Let $X=(x_2,\ldots,x_k)$ be a forget node, whose child is $X'=(x_1,\ldots,x_k)$.
At $X$, we need to determine how to block edges connecting $x_1$ and the rest $x_2,\ldots,x_k$. There are at most $k-1$ edges to block so we simply go over at most $2^{k-1}$ blocking options.
For each specific blocking option (corresponding to a spending of $b''$), we convert a tuple in $DP(X')$ to a tuple in $DP(X)$ as follows (the new tuple is discarded if $b'+b''>b$):
\begin{footnotesize}
\[
    \left(\begin{bmatrix}
        d_{11} & \ldots & d_{1k} \\
        \ldots\\
        d_{k1} & \ldots & d_{kk} \\
    \end{bmatrix},b'\right)
    \rightarrow
    \left(\begin{bmatrix}
        d_{22}' & \ldots & d_{2k}'\\
        \ldots\\
        d_{k2}' & \ldots & d_{kk}'\\
    \end{bmatrix},b'+b''\right)
\]
\end{footnotesize}

The $d_{ij}'$ are updated distances considering the newly put back edges.
We need to run an {\em all-pair shortest path} algorithm with complexity $O(k^3)$ for this update.

\noindent{\bf Join node:} Let $X$ be a join node with two children $X_1$ and
$X_2$.  For $(M_1,b_1)\in DP(X_1)$ and $(M_2,b_2)\in DP(X_2)$, we label $(M',
b_1+b_2)$ as a possible tuple in $DP(X)$ if $b_1+b_2\le b$. $M'$ is
the element-wise minimum between $M_1$ and $M_2$.

\begin{theorem}
    \label{thm:tdcycle}
    {\sc TDCycle}'s complexity is $O(H^{2w^2}b^2w^2n)$.
\end{theorem}

Proof is deferred to the appendix due to space constraint.

\vspace{.05in}
To summarize our dynamic program, we follow a bottom-up order
(from leaf nodes of nice tree decomposition {\sc TD} to the root).  We propagate the set of all possible tuples $(M,b')$ as we process the nodes. At introduce/join nodes, we follow a pre-determined propagation rule and do not make any blocking decisions. At forget nodes, we decide which edges to block from at most $w$ edges. Given a specific AD graph and its corresponding tree decomposition {\sc TD}, we can convert our dynamic program to a reinforcement learning environment as follows, which leads to an anytime algorithm that scales better than dynamic program (i.e., RL can always produce a solution, which may or may not be optimal, and generally improves over time; on the other hand, dynamic program cannot scale to handle slightly larger tree widths).
{\em Our conversion technique can potentially be applied to tree decomposition based
dynamic programs for other combinatorial optimisation problems.}

\begin{itemize}

    \item We use post-order traversal on {\sc TD} to create an ordering of the nodes (children are processed before parents).

    \item Instead of propagating {\em all possible} tuples $(M,b')$, we
        only
        propagate the {\em best tuple} (we have found so far during RL training). For example, consider a forget
        node $X$ with child $X'$.  The best tuple at node $X'$ is passed on to $X$ as {\em observation}.
        The {\em action} for this observation is then to decide which edges to block at $X$.

    \item Specifically to our model, for forget node $X$, if there are $k$ ($k\le w$) edges to block,
        then we treat it as $k$ separate steps in our reinforcement learning environment. That is, every {\em step} makes a blocking decision on a single edge and the action space is always binary.

    \item For introduce/join nodes, since we do not need to make any decisions, our reinforcement learning environment
        automatically processes these nodes (between {\em steps}).

    \item We set a final {\em reward} that is equal to the solution quality.

\end{itemize}

After we convert a specific AD graph into a reinforcement learning environment,
we can then apply standard RL algorithms
to search for the optimal blocking policy for the AD graph under discussion.

Following our conversion, both the observation and the action spaces are small when
tree widths are small.
Unfortunately, one downside of the above conversion technique is that there is
no guarantee on having a small {\em episode length}. We could argue that it is
impossible to guarantee small observation space, small action space, and small
episode length at the same time, unless the AD graph is relatively small in
scale. {\em After all, we are solving NP-hard problems.} Experimentally, for
smaller AD graphs, we are able to achieve near-optimal performances as the
episode lengths are manageable.  For larger AD graphs, the episode lengths are
too long. We introduce the following heuristic for limiting the episode length.
Let $T$ be the target episode length.  Our idea is to hand-pick $T$ {\em
relatively important} edges and set the unpicked edges not blockable.
In our experiments, we first calculate the {\em min cut} that separates the entry
nodes from {\sc DA} (unblockable edges' capacities are set to be large).  Let the number of blockable edges in the min cut be $C$.
If $C\ge T$, then we simply treat these $C$ edges as important and set the episode
length to $C$. If $C<T$, then we add in $T-C$ blockable edges that are
closest to {\sc DA}.

In the appendix, we include an experiment showing that our RL-based approach is not merely performing ``random searching'' via exploration. It is indeed capable of ``learning'' to react to the given observation.
Under our original approach, the observation contains the distance matrix $M$,
the budget spent $b'$, and also the current step index.
We show that if we replace $M$ by the zero matrix or by a random matrix, then the training results significantly downgrade.

\section{Kernelization}

{\em Kernelization} is a commonly used fixed-parameter analysis technique that
preprocesses a given problem instance and converts it to a much smaller {\em
equivalent} problem, called the {\em kernel}. We require that the kernel's size
be bounded by the special parameters (and not depend on $n$).

As mentioned in the introduction, for practical AD graphs,
most nodes have at most one out-going edge.
If an edge is not useful for the attacker,
then we can remove it without loss of generality. If an edge is useful
for the attacker, then generally, it is {\em privilege escalating}. It is rare
for a node to have two separate {\em privilege escalating} out-going edges (as they often correspond to security exceptions or misconfigurations).
We use {\sc SPLIT} to denote the set of all splitting nodes (nodes with multiple out-going edges).
We use {\sc SPLIT+DA} to denote {\sc SPLIT} with {\sc DA} added.
We use {\sc ENTRY} to denote the set of all entry nodes.
We invent a new parameter called the number of {\em non-splitting paths}.
Experimentally, this parameter leads to algorithms that scale
exceptionally well for synthetic AD graphs generated using two different open
source AD graph generators.

\begin{definition}[Non-splitting path]
    Given node $u$, let $v$
    be one of $u$'s successors. The non-splitting path {\sc nsp}$(u,v)$ is defined
    recursively:

    \begin{itemize}

        \item If $v\in\textsc{SPLIT+DA}$, then $\textsc{nsp}(u,v)$
            is $u\rightarrow v$.

        \item Otherwise, $v$ must have a unique successor $v'$.
            $\textsc{nsp}(u,v)$ is the path that combines $u\rightarrow v$ and
            $\textsc{nsp}(v, v')$.

    \end{itemize}

\end{definition}

    In words, $\textsc{nsp}(u,v)$ is the path that goes from $u$ to $v$, then
    repeatedly moves onto the only successor of $v$ if $v$ has only one successor,
    until we reach either a splitting node or {\sc DA}.
    We use $\textsc{dest}(u,v)$ to denote the ending node of $\textsc{nsp}(u,v)$.
    We have $\textsc{dest}(u,v)\in \textsc{SPLIT+DA}$.
    A non-splitting path is called {\em blockable} if at least one of its edges
    is blockable.

    An AD graph can be viewed as the {\bf joint} of the following set of non-splitting
    paths. Our parameter (the number of non-splitting paths {\sc \#NSP}) is the size of this set.
    \[\{\textsc{nsp}(u,v) | u\in \textsc{SPLIT}\cup \textsc{ENTRY}, v\in \textsc{Successors}(u)\}\]

    The blockable edge furthest away from $u$ on the path
    $\textsc{nsp}(u,v)$ is denoted as $\textsc{bw}(u,v)$.
{\sc bw} stands for {\em block-worthy} due to the following lemma.


\begin{lemma}\label{lm:bw} For both pure and mixed strategy blocking, we never need to spend
    more than one unit of budget on a non-splitting path.
    For any problem instance, there
exists an optimal defense that blocks only edges
    from the following set:
    \[\textsc{BW}=\{\textsc{bw}(u,v)|u\in \textsc{SPLIT}\cup \textsc{ENTRY}, v\in \textsc{Successors}(u)\}\]
\end{lemma}

We present how to formulate our model as a nonlinear program, based on
the aforementioned non-splitting path interpretation.
The nonlinear program can then be converted to MIPs and be efficiently solved using state-of-the-art MIP solvers.
We use $B_e$ to denote the unit of budget spent on edge $e$.  $B_e$ is binary
for pure strategy blocking and is between $0$ and $1$ for mixed strategy
blocking.  As mentioned earlier, $B_e\ge 0$ for $e\in \textsc{BW}$ and
$B_e=0$ for $e\notin \textsc{BW}$.
We use $r_u$ to denote the
success rate of node $u$.  The success rate of a node is the success rate of
the optimal attack path starting from this node, under the current
defense (i.e., the $B_e$).
We use $c_{u, v}$ to denote the
success rate of the non-splitting path $\textsc{nsp}(u,v)$ when no blocking
is applied. $c_{u,v}=\prod_{e\in\textsc{nsp}(u,v)}(1-f(e))$.
The $c_{u,v}$ are constants.
We use $r^*$ to represent the attacker's optimal success rate.
\[C=\{(u,v)|u\in \textsc{SPLIT}\cup \textsc{ENTRY}, v\in \textsc{Successors}(u)\}\]
\[C^+=\{(u,v)|(u,v)\in C, \textsc{nsp}(u,v) \text{ blockable}\}\]
\[C^-=\{(u,v)|(u,v)\in C, \textsc{nsp}(u,v) \text{ not blockable}\}\]

We have the following nonlinear program:
\begin{equation*}
\begin{array}{rclcl}
\min & \multicolumn{3}{l}{r^*} \\
    r^* & \geq & r_u && \forall u\in \textsc{ENTRY} \\
    r_u & \geq & r_{\textsc{dest}(u,v)}\cdot c_{u,v}\cdot (1-B_{\textsc{bw}(u,v)}) \hspace{-.2in}&& \forall (u,v)\in C^+ \\
    r_u & \geq & r_{\textsc{dest}(u,v)}\cdot c_{u,v} && \forall (u,v)\in C^-\\
    b & \geq & \sum_{e \in \textsc{BW}}B_e & & \\
    r_{\textsc{DA}} & = & 1 && \\
    r_u, r^* & \in & [0,1] & & \\
    B_e & \in & \{0,1\}\text{ or }[0,1]& & \\
\end{array}
\end{equation*}


The above program has at most $O(\textsc{\#NSP})$ variables and
at most $O(\textsc{\#NSP})$ constraints (both do not depend on $n$).

\vspace{.1in}
\noindent{\bf Integer program for pure strategy blocking:}
For pure strategy blocking, the above program can be converted to an IP by
rewriting $r_u\ge r_{\textsc{dest}(u,v)} \cdot c_{u,v} \cdot
(1-B_{\textsc{BW}(u,v)})$ into an equivalent linear form $r_u\ge
r_{\textsc{dest}(u,v)} \cdot c_{u,v}-B_{\textsc{BW}(u,v)}$.


\section{Mixed Strategy Blocking}

We can convert the nonlinear program from the previous section into a
program that is {\em almost} linear.  We first observe that if under the
optimal mixed strategy defense, the attacker's success rate is at least
$\epsilon$, then that means we never want to block any edge with a probability
that is {\em strictly} more than $1-\epsilon$.  Due to this observation, we artificially enforce that
we do not block an edge with more than $1-\epsilon$ probability and solve for the
optimal defense under this restriction. If in the end result, the attacker's
success rate is at least $\epsilon$, then our restriction is not actually a restriction
at all. If our solution says that the attacker's success rate is less than
$\epsilon$, then we have an almost optimal defense anyway.
In this paper, we set $\epsilon=0.01$.
That is, for $e\in\textsc{BW}$, we
require that $0\le B_e\le 0.99$ instead of $0\le B_e\le 1$.

With the above observation, we can convert the nonlinear program, which
involves multiplication, to an {\em almost} linear program as follows.  Our
trick is to replace $r_u$ by $r'_u=-\ln(r_u)$, replace $r^*$ by
$r'^*=-\ln(r^*)$, replace $B_e$ by $B'_e=-\ln(1-B_e)$, and finally replace
$c_{u,v}$ by $c'_{u,v}=-\ln(c_{u,v})$.
Our variables are now $r'^*$, the $r'_u$ and the $B'_e$.
Due to monotonicity of natural log,
we can rewrite the earlier nonlinear program as:
\begin{equation*}
\begin{array}{rclcl}
\max & \multicolumn{3}{l}{r'^*} \\
    r'^* & \leq & r'_u && \forall u\in \textsc{ENTRY} \\
    r'_u & \leq & r'_{\textsc{dest}(u,v)}+c'_{u,v}+B'_{\textsc{bw}(u,v)}\hspace{-.2in} && \forall (u,v)\in C^+\\
    r'_u & \leq & r'_{\textsc{dest}(u,v)}+c'_{u,v} && \forall (u,v)\in C^-\\
    b & \geq & \sum_{e \in \textsc{BW}}(1-e^{-B'_e}) & & \\
    r'_{\textsc{DA}} & = & 0 & & \\
    r'_u, r'^* & \in & [0,\infty) & & \\
    B'_e & \in & [0,-\ln(0.01)] & &
\end{array}
\end{equation*}


Unfortunately, the above program is not linear as the budget constraint
is not linear. Furthermore, the above program is not
even convex: Since $1-e^{-x}$ is concave, the average of two feasible solutions
may violate the budget constraint.

\vspace{-.05in}
\subsection{Iterative LP based approximation heuristic}

We note that $1-e^{-x}$ is increasing. That is, as long as we push down the total of $B'_e$, we eventually will reach a situation where the budget constraint is satisfied.
That is, we rewrite the budget constraint as $b'\ge \sum_{e\in\textsc{BW}}B'_e$, which
is a linear constraint. We {\bf guess} a value for $b'$ and then solve the corresponding LP.  We verify whether the LP solution also satisfies the original nonlinear
budget constraint. If it does, then that means we could increase our guess of $b'$
(increasing $b'$ means spending more budget).
If our LP solution violates the original nonlinear budget constraint, then it
means we should decrease our guess of $b'$.  A good guess of $b'$ can be obtained via a
binary search from $0$ to $-|\textsc{BW}|\ln(0.01)$.

\vspace{-.05in}
\subsection{MIP based approximation}

Another way to address the nonlinear budget constraint is to add the $B_e$ back
into the model ($0\le B_e\le 0.99$ for $e\in\textsc{BW}$). The budget
constraint $\sum_{e\in\textsc{BW}}B_e\le b$ is now back to linear.  Of course,
this cannot be the end of the story, since we also need to link the $B_e$ and
the $B'_e$ together.  We essentially have introduced a new set of nonlinear
constraints, which are $B'_e = -\ln(1-B_e)$ for $e\in \textsc{BW}$.

The function $-\ln(1-x)$ is close to a straight line if we focus on a small
interval.
If $x\in [a,b]$, the straight line
connecting $(a,-\ln(1-a))$ and $(b,-\ln(1-b))$ is an upper bound of
$-\ln(1-x)$.
A straight line representing a lower bound of $-\ln(1-x)$ is the tangent line
at the interval's mid point $\frac{a+b}{2}$.
An illustration of $-\ln(1-x)$, bounded above and below by straight lines
is provided in the appendix.

Given a specific $B_e$, we could divide $B_e$'s region $[0,0.99]$ into multiple
smaller intervals.
For each region, we have two
straight lines that represent the lower and upper bounds on $-\ln(1-x)$.  We
could join the regions and compose $g_L(x)$ and $g_U(x)$, which are the {\em
piece-wise linear} lower and upper bounds on $-\ln(1-x)$.  We replace
$B'_e=-\ln(1-B_e)$ by $B'_e=g_L(B_e)$ or $B'_e=g_U(B_e)$, respectively, to
create two different MIP programs.
We use the {\em multiple choice model} presented in
\cite{Croxton2003:Comparison} to implement piece-wise linear constraints with the
help of auxiliary binary variables.
The program with $B'_e=g_L(B_e)$ underestimates
the $B'_e$, which results in an overestimation of the budget spending and
therefore an overestimation of the attacker's success rate.  This program
results in {\bf an achieved feasible defense}.  The program with $B'_e=g_U(B_e)$
on the other hand
results in {\bf a lower bound on the attacker's success rate}.  When the intervals
are fine enough,
experimentally, the achieved feasible
defense is close to the lower bound (therefore close to optimality).


\section{Experiments}
All our experiments are carried out on a desktop with i7-12700 CPU and NVIDIA GeForce RTX 3070 GPU.
Our MIP solver is Gurobi 9.5.1.
For pure strategy blocking, we proposed
two algorithms: {\sc TDCycle} and {\sc IP} (integer program based on kernelization).  We
also include a third algorithm {\sc Greedy} to serve as a baseline.
{\sc Greedy} spends one unit of budget in each round for a total of $b$ rounds.
In each round, it greedily blocks one edge to maximally
decrease the attacker's success rate.  For mixed strategy blocking, we have an
iterative LP based heuristic {\sc IterLP}, a mixed integer program {\sc
MIP-F(easible)} for generating a feasible defense and a mixed integer program
{\sc MIP-LB} for generating a lower bound on the attacker's success rate.

We evaluate our algorithms using two attack graphs generated using {\sc BloodHound} team's synthetic graph generator {\sc
DBCreator}. We call the attack graph {\sc R2000} and {\sc R4000}, which are obtained by setting the
number of computers in the AD environment to $2000$ and $4000$.
{\sc R2000} contains $5997$ nodes and $18795$ edges and
{\sc R4000} contains $12001$ nodes and $45780$ edges.
We also generate a third attack graph using a different open source
synthetic graph generator {\sc adsimulator}. {\sc adsimulator} by
default generates a trivially small graph.
We increase all its default parameters by a factor of $10$ and create an
attack graph called {\sc ADS10}. {\sc ADS10} contains $3015$ nodes
and $12775$ edges.
Even though {\sc ADS10} contains less nodes, experimentally it is actually more
expensive to work with compared to {\sc R2000} and {\sc R4000}, as it is further away from a tree.

We only consider three edge types: {\sc AdminTo}, {\sc MemberOf},
and {\sc HasSession}.
These are a
   representative sample of edges types used in BloodHound.
We set the failure rates of all edges of type {\sc HasSession} to $0.2$ (requiring
a session between an account and a computer, therefore more likely to fail)
and set the failure rates of all edges of the other two types to $0.05$.
We set the number of entry nodes to $20$. We select $40$ nodes that are
furthest away from {\sc DA} (in terms of the number of hops to reach {\sc DA})
and randomly draw $20$ nodes among them to be the entry nodes.  We define
$\textsc{Hop}(e)$ to be the minimum number of hops between an edge $e$ and {\sc
DA}.  We set {\sc MaxHop} to be the maximum value for $\textsc{Hop}(e)$.  An
edge is set to be blockable with probability
$\frac{\textsc{Hop}(e)}{\textsc{MaxHop}}$.  That is, edges further away from
{\sc DA} are set to be more likely to be blockable. Generally speaking, edges further
away from {\sc DA} tend to be about individual
employees' accesses instead of accesses between servers and admins.
We set the budget to $5$ and $10$.
All experiments are repeated $10$ times, with different random draws
of the entry nodes and the blockable edges. The numbers in the table
are the attacker's average success rates over $10$ trials. The numbers in the parenthesis are the average running time.


\begin{footnotesize}
\begin{center}
    \begin{tabular}{ |l|l|l|l| }
 \hline
        budget=5 & {\sc R2000} & {\sc R4000} & {\sc ADS10} \\
 \hline
        {\sc Greedy} & $0.521$ ($0.04$s) & $0.376$ ($0.34$s) & $0.448$ ($4.57$s)\\
 \hline
        {\sc TDCycle} & $0.480$ ($0.10$s) & $0.373$ ($15366$s) & -  \\
 \hline
        {\sc IP} & $0.480$ ($0.01$s) & $0.373$ ($0.06$s) & $0.409$ ($0.09$s)\\
 \hline
 \hline
        {\sc IterLP} & $0.337$ ($0.37$s) & $0.180$ ($0.57$s) & $0.300$ ($1.90$s)\\
 \hline
        {\sc MIP-F} & $0.335$ ($0.11$s) & $0.179$ ($3.54$s) & $0.303$ ($4.75$s)\\
 \hline
        {\sc MIP-LB}  & $0.333$ ($0.11$s) & $0.176$ ($2.61$s) & $0.297$ ($4.98$s) \\
 \hline
\end{tabular}
\end{center}
\end{footnotesize}

\begin{footnotesize}
\begin{center}
    \begin{tabular}{ |l|l|l|l| }
 \hline
        budget=10 & {\sc R2000} & {\sc R4000} & {\sc ADS10} \\
 \hline
        {\sc Greedy} & $0.499$ ($0.07$s) & $0.376$ ($0.65$s) & $0.448$ ($8.97$s)\\
 \hline
        {\sc TDCycle} & $0.263$ ($0.18$s) & - & -  \\
 \hline
        {\sc IP} & $0.263$ ($0.01$s) & $0.117$ ($0.03$s) & $0.315$ ($0.10$s)\\
 \hline
 \hline
        {\sc IterLP} & $0.266$ ($0.39$s) & $0.033$ ($0.60$s) & $0.190$ ($1.92$s)\\
 \hline
        {\sc MIP-F} & $0.270$ ($0.02$s) & $0.023$ ($0.54$s) & $0.191$ ($11.27$s)\\
 \hline
        {\sc MIP-LB}  & $0.263$ ($0.02$s) & $0.014$ ($0.24$s) & $0.183$ ($3.62$s) \\
 \hline
\end{tabular}
\end{center}
\end{footnotesize}

\noindent{\bf Interpretation of Results:}
For pure strategy blocking, {\sc TDCycle} and {\sc IP} are both expected to
produce the optimal results. As expected, they perform better than {\sc Greedy}. {\sc TDCycle} doesn't scale for 3 out of 6 settings.
On the other hand, {\sc IP} scales exceptionally well.
As mentioned earlier, {\sc IP} scales better since it is based on a parameter
that we invent specifically for describing AD graphs.
For graphs with large number of non-splitting paths and small tree widths,
we expect {\sc TDCycle} to scale better, {\em but such graphs may not be AD graphs}.
For mixed strategy blocking, the attacker's success rates under both {\sc IterLP} and {\sc MIP-F(easible)} are close to {\sc MIP-LB} (lower bound on the attacker's success rate), which indicates that both heuristics are near-optimal.

\vspace{.05in}
\noindent{\bf Results on scaling {\sc TDCycle} via reinforcement learning:}
We present results on two settings: {\sc R2000} with $b=5$ and {\sc ADS10} with $b=10$.
These two are the {\em cheapest} and the {\em most expensive} among our six experimental settings.
All our experimental setups are the same as before. We directly apply {\em Proximal Policy Optimization} {\sc PPO} \cite{Schulman17:Proximal}.\footnote{We list our hyper-parameters in the appendix.}
For each environment, we use training seed $0$ to $9$ and record the best result.

\begin{footnotesize}
\begin{center}
    \begin{tabular}{ |l|l|l|l|l|l| }
 \hline
         & Opt & Greedy & RL & $=$Opt & Time \\
 \hline
        {\sc R2000}, $b=5$ & $0.480$ & $0.521$ & $0.480$ & $10/10$ & 1hr \\
 \hline
        {\sc ADS10}, $b=10$ & $0.315$ & $0.448$ & $0.319$ & $9/10$ & 4hr \\
 \hline
\end{tabular}
\end{center}
\end{footnotesize}

{\sc Opt}, {\sc Greedy}, {\sc RL} each represents the average performance of
these three different approaches (average over $10$ trials with random draws of the entry nodes and the blockable edges). ``Time'' refers to
training time.
For {\sc R2000} with $b=5$, we obtain the optimal result in $10$ out of $10$ trials
without limiting the episode length. The maximum episode length is $27$ over $10$ trials.
For {\sc ADS10} with $b=10$, we set an episode length of $15$ and obtain the optimal result in $9$ out of $10$ trials.
We recall that, for this setting, {\sc TDCycle} doesn't scale at all, but we manage to achieve near-optimal results via reinforcement learning.

\newpage

\section*{Acknowledgements}
Frank Neumann has been supported by the Australian Research Council through grant FT200100536. Hung Nguyen is supported by the Next Generation Technology Fund
   (NGTF-Cyber) grant  MyIP 10614 and the Australian Research Council
   (ARC-NISDRG) grant NI210100139.
This work was supported with supercomputing resources provided by the Phoenix HPC service at the University of Adelaide.

\bibliographystyle{unsrtnat}
\bibliography{/home/mingyu/Dropbox/nixos/mg.bib,/home/mingyu/Dropbox/nixos/mggame.bib,/home/mingyu/Dropbox/nixos/extra.bib}

\newpage

\section*{Appendix}

\subsection{An example synthetic AD graph}

\begin{figure}[h]
    \centering
  \includegraphics[width=0.5\linewidth]{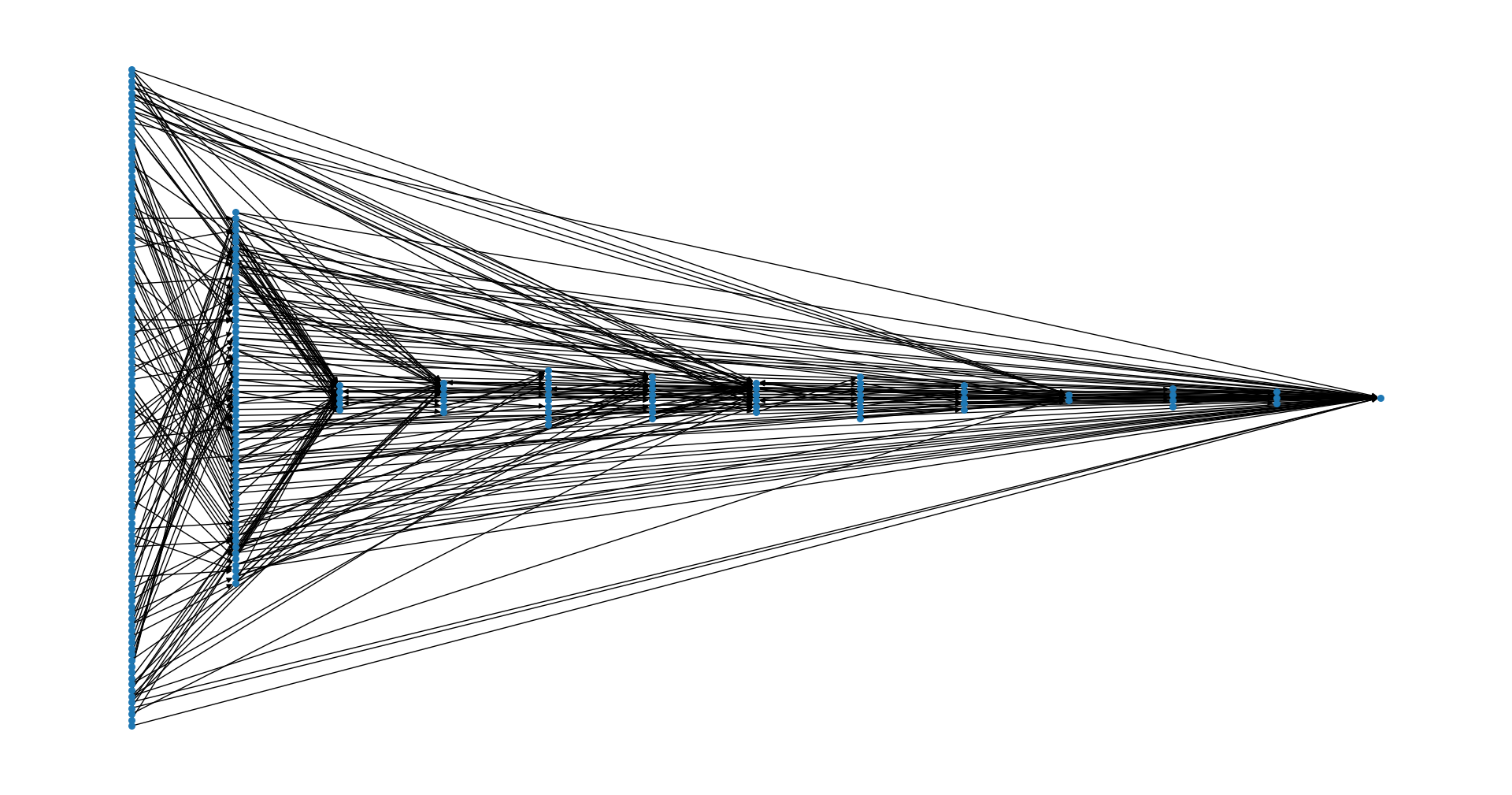}
    \caption{An example Active Directory attack graph generated using {\sc BloodHound} team's synthetic graph generator {\sc DBCreator} (setting the number of computers in the network to $1000$). Only nodes reachable to {\sc DA} (rightmost node) are shown. Every node only needs several hops to reach {\sc DA} and the whole graph is similar to a tree.}
  \label{fig:r1000}
\end{figure}

\subsection{How to optimally block $2$ edges from the example attack graph described in Figure~\ref{fig:example}}


    \begin{itemize}

        \item Optimal pure strategy defense:
One optimal pure strategy defense is to block $2\rightarrow 1$ and $5\rightarrow 1$.
After these two edges are blocked, the attacker's optimal attack path is
$4\rightarrow 1\rightarrow 0$, with a success rate $(1-0.2)(1-0.05)=0.76$.

        \item Optimal mixed strategy defense:
The optimal
mixed strategy defense is to block $2\rightarrow 1$ with probability
$0.675$, block $4\rightarrow 1$ with probability $0.634$, and block
    $5\rightarrow 1$ with probability $0.691$.  Note that $0.675+0.634+0.691=2$.  Under
this defense, there are three optimal attack paths for the attacker,
which are $3\rightarrow 2\rightarrow 1\rightarrow 0$, $4\rightarrow
1\rightarrow 0$, and $5\rightarrow 1\rightarrow 0$. All three optimal attack
paths' success rates are $0.279$.  For example, for attack path $3\rightarrow
2\rightarrow 1\rightarrow 0$, the success rate equals
$(1-0.05)(1-0.05)(1-0.675)(1-0.05)=0.279$.

    \end{itemize}

\subsection{Proof of Theorem~\ref{thm:nphard}}

We only present the proof for mixed strategy blocking.  For pure strategy, we
can derive a similar proof or directly apply the result from
\cite{Baier2006:LengthBounded}, where the authors show that for directed
graphs, the {\em single-source single-destination 4-length-bounded edge cut} problem
is NP-hard.  4-length-bounded edge cut is to calculate the minimum
number of edges to block so that the distance from the source to the
destination is at least 5 (equivalently, exactly 5).

\begin{proof}[Proof for mixed strategy only]

\begin{figure}[h]
    \centering
  \includegraphics[width=0.5\linewidth]{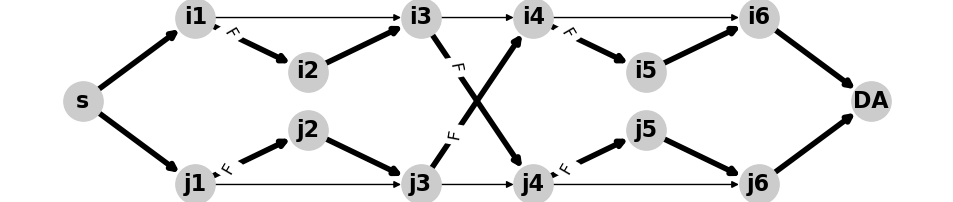}
    \caption{Proof gadget: $i1$ to $i6$ and $j1$ to $j6$ represent {\em connected} node $i$ and $j$ from a vertex cover instance}
  \label{fig:nphard}
\end{figure}

\cite{Dinur2005:Hardness}
    showed that the minimum vertex cover problem is NP-hard to approximate to within a
    factor of $1.3606$.  We show that optimal mixed strategy blocking
    is at least as hard as approximating vertex cover to within
    a factor of $1.3606$ by reducing vertex cover to mixed strategy blocking.


Let $G(V',E')$ be a vertex cover instance. $G'$ is an undirected graph.  Let
    $n'=|V'|$ and let $k^*$ be the minimum vertex cover number.  We construct
    an attack graph as follows.  We first construct a single entry node $s$ and
    a single destination {\sc DA}.  For every node $i \in V'$, we construct $6$
    nodes and $9$ edges in the attack graph. For $i \in V'$, the $6$ nodes are
    $i1$ to $i6$, and the $9$ edges are $s\rightarrow i1$, $i6\rightarrow
    \textsc{DA}$, and the $7$ edges among $i1$ to $i6$ (as shown in
    Figure~\ref{fig:nphard}).
    For every edge $(i,j)\in E'$, we construct
    $2$ edges in the attack graph. If $i$ and $j$ are connected by an edge in $G'$,
    we construct $i3\rightarrow j4$ and $j3\rightarrow i4$.
    The thick edges are not blockable and the thin edges are blockable. Edges marked with
    the label {\bf F} are called {\em broken edges}. Broken edges have a failure rate of $1-\epsilon^4$,
    where $\epsilon=\frac{1}{n'}$.
    Every other edge has a failure rate of $0$.

    We define $b^*$ to be the minimum budget where under the optimal mixed
    strategy defense,
    the attacker's success rate is at most $\epsilon^8$.
    We first show that $b^*\le n'+k^*$. Assuming that
    an oracle provides us with a minimum vertex cover solution with $k^*$ vertices, we can convert this vertex cover
    solution to a pure strategy defense that uses a budget of $n'+k^*$
    and can ensure that the attacker's success rate is at most $\epsilon^8$.
    Suppose
    $i$ and $j$ are connected and only $i$ belongs to the minimum vertex cover,
    then for $i$ we spend two units of budget (block $i1\rightarrow i3$ and
    $i4\rightarrow i6$) and for $j$ we spend one unit of budget (only block
    $j3\rightarrow j4$).  With the above blocking, all attack
    paths must pass through at least two broken edges. Therefore, by blocking
    $2\cdot k^* + 1\cdot (n'-k^*) =n'+k^*$ edges, the attacker's success rate is at most $\epsilon^8$. Mixed strategy is more general than pure strategy. By the definition
    of $b^*$, we must have $b^*\le n'+k^*$.

    Now let us consider the optimal mixed strategy blocking solution
    corresponding to budget $b^*$.
    We show that we can convert this optimal
    mixed strategy solution to a pure strategy solution by increasing the
    budget spending by at most $\frac{1}{1-\epsilon^2}$.  The resulting pure
    strategy blocking solution corresponds to a valid vertex cover solution
    that is arbitrarily close to optimality. This essentially means that if we
    can solve the mixed strategy blocking problem, then we can approximate
    vertex cover to within a factor that is arbitrarily close to $1$, which
    shows optimal mixed strategy blocking is at least as hard as the known
    NP-hard problem of approximating
    vertex cover to within a factor of $1.3606$.

    Under the optimal mixed strategy defense with budget $b^*$, we call an edge {\em nearly blocked} if the
    defender blocks it with at least $1-\epsilon^2$ probability. For node $i$,
    we define $L(i)$ to be $1$ if the {\em left} edge $i1\rightarrow i3$ is
    nearly blocked. $L(i)=0$ otherwise.
    Similarly, we define $M(i)$ to be $1$ if the {\em middle} edge $i3\rightarrow i4$ is
    nearly blocked and $R(i)$ to be $1$ if the {\em right} edge $i4\rightarrow i6$ is
    nearly blocked.
    For two connected nodes $i$ and $j$ in $G'$, we can prove six
    inequalities:
    1) $L(i)+R(j)\ge 1$;
    2) $L(j)+R(i)\ge 1$;
    3) $L(i)+M(i)\ge 1$;
    4) $M(i)+R(i)\ge 1$;
    5) $L(j)+M(j)\ge 1$;
    6) $M(j)+R(j)\ge 1$.
    Let us take $L(i)+R(j)\ge 1$ as an example. If neither $L(i)$ nor $R(j)$
    is nearly blocked, then the success rate by going through $s\rightarrow i1\rightarrow i3\rightarrow j4\rightarrow j6\rightarrow \textsc{DA}$ is more than $\epsilon^8$.
    The above six inequalities imply two things.  First,
    $L(i)+M(i)+R(i)+L(j)+M(j)+R(j)\ge 3$.  That is, for $i$ and $j$ that are
    connected in $G'$, at least one of $i$ and $j$ have at least two edges nearly blocked.
    If node $i$ has at least two edges
    that are nearly blocked, then we can modify the optimal mixed strategy and change to fully block $i1\rightarrow i3$
    and $i4\rightarrow i6$ instead (budget on other edges reduced to $0$). This way, the attacker can never gain a success
    rate that is strictly higher than $\epsilon^8$ by using any of $i$'s edges. By
    changing from nearly blocked to fully blocked, we increase the budget spending by
    at most $\frac{1}{1-\epsilon^2}$.  For node $j$,
    due to $L(j)+M(j)\ge 1$ and $M(j)+R(j)\ge 1$, we have that
    at least one edge corresponding to $j$ is nearly blocked and if there is exactly one edge nearly
    blocked, then that edge must be the middle one $j3\rightarrow j4$.
    We again can increase our budget spending
    by at most $\frac{1}{1-\epsilon^2}$ by changing the middle edge from nearly
    blocked to fully blocked (and reduce the spending on the other two edges to $0$).
    Essentially, given the optimal mixed strategy defense, we can increase
    the budget spending by at most $\frac{1}{1-\epsilon^2}$ in order to convert
    to a pure strategy defense that corresponds to a
    feasible vertex cover solution. Thus, we have
     $\frac{b^*}{1-\epsilon^2}\ge n'+k^*$ and
     $\frac{b^*}{1-\epsilon^2}-n'$ is an achieved vertex cover number.
     Earlier we proved that $b^*-n'$ is a lower bound of the vertex cover number.
     The approximation ratio is at most
     $\frac{\frac{b^*}{1-\epsilon^2}-n'}{b^*-n'}$, which becomes strictly
     less than $1.3606$ when
     $n'$ reaches a constant threshold.
     The mathematical details are presented below.
     For small $n'$ values that are below the constant threshold, we could brute force
     to calculate the optimal vertex cover solution, which is also
     within a factor of $1.3606$. That is, optimal mixed strategy blocking
     is at least as hard as the known NP-hard problem of approximating vertex cover to within a factor of $1.3606$.

     Below we show the details on the constant threshold:
     \[\frac{\frac{b^*}{1-\epsilon^2}-n'}{b^*-n'}=\frac{\frac{b^*}{1-\frac{1}{n'^2}}-n'}{b^*-n'}=\frac{\frac{b^*n'^2}{n'^2-1}-n'}{b^*-n'}\]
     \[=\frac{b^*n'^2-n'(n'^2-1)}{(b^*-n')(n'^2-1)}=\frac{b^*n'^2-n'^3+n'}{b^*n'^2-n'^3+n'-b^*}\]
     \[=1+\frac{b^*}{b^*n'^2-n'^3+n'-b^*}\]
     \[=1+\frac{b^*}{(b^*-n')(n'^2-1)}\]

     The above expression is monotone in $b^*$
     and is maximized when $b^*$ takes the minimum possible value.
     We know that $b^*$ is at least $(1-\epsilon^2)(n'+k^*)\ge (1-\frac{1}{n'^2})(n'+1)$.
     ($k^*$ is at least $1$ as it is the vertex cover number.)
     The approximation ratio is then at most
     \[1+\frac{(1-\frac{1}{n'^2})(n'+1)}{((1-\frac{1}{n'^2})(n'+1)-n')(n'^2-1)}\]
     \[=1+\frac{(n'^2-1)(n'+1)}{((n'^2-1)(n'+1)-n'^3)(n'^2-1)}\]
     \[=1+\frac{n'+1}{(n'^2-1)(n'+1)-n'^3}\]
     \[=1+\frac{n'+1}{n'^2-n'-1}\]

     The above expression approaches $1$ as $n'$ approaches infinity.
     There exists a constant threshold where the above is at most $1.3606$
     when $n'$ is above the threshold.

\end{proof}



\subsection{Tree decomposition background}

\begin{definition}[Tree Decomposition and Tree Width] Let $G(V,E)$ be an undirected graph.
    Let $T$ be a tree with $t$ tree nodes ($T_1,T_2,\ldots,T_t$), where every tree node is a bag of graph vertices (subset of $V$).
    $T$ is a {\em tree decomposition} of $G$ iff

    \begin{itemize}

        \item The union of the $T_i$ equals $V$.

        \item For every edge $e\in E$, there exists at least one $T_i$ such that both ends of $e$ are in $T_i$.

        \item For every vertex $v\in V$, let $T(v)$ be the set of tree nodes containing
            $v$. The subgraph of $T$ induced by $T(v)$ must be connected (must form a
            tree).

    \end{itemize}

    The {\em tree width} of $G$ (under the tree decomposition $T$) equals
    the maximum size/cardinality among the $T_i$, then minus $1$.

\end{definition}

AD graphs are directed, but we treat AD graphs as undirected for the purpose
of applying tree decomposition. For all AD graphs used
in this paper (generated using {\sc DBCreator} and {\sc adsimulator}), between any pair of directly connected vertices $u$ and $v$,
either $u\rightarrow v$ or $v\rightarrow u$ exists (never both).
Hence treating the graphs as undirected does cause any ambiguity or lose
any information,
as we can always refer back to the original AD graph to query the edge directions.
(Even if
both $u\rightarrow v$ and $v\rightarrow u$ exist,
we simply need to add one auxiliary vertex in between to remove any ambiguity caused.
For example, keep $u\rightarrow v$ and add $x$ in between $v\rightarrow u$
so that it becomes $v\rightarrow x\rightarrow u$.)

The optimal tree decomposition with the minimum tree width is NP-hard to
compute~\cite{Arnborg1987:Complexity}. In our experiments, we adopt the {\em
vertex elimination} heuristic for generating tree
decomposition~\cite{Bodlaender2006:Exact}.
This heuristic maps a permutation of the graph nodes into a tree decomposition.
The pseudocode is included.

\begin{algorithm}[h!]
\caption*{Vertex elimination heuristic for tree decomposition}
\textbf{Input}: An undirected graph $G$ and an arbitrary permutation of the graph vertices
\begin{algorithmic}[1]
\FOR {vertex $i$ in $G$ ordered by the permutation}
    \STATE let $N(i)$ be $i$'s neighbours under the current graph
    \STATE create tree node $T_i = \{i\}\cup N(i)$
    \FOR {every pair of vertices $a,b\in N(i)$}
        \STATE add edge $(a,b)$ to $G$
    \ENDFOR
    \STATE remove $i$ from $G$
\ENDFOR
\FOR {$T_i$ from $1$ to $n$}
    \STATE let $j=\max\{ T_i/\{i\}\}$
    \STATE connect tree node $T_i$ to $T_j$
\ENDFOR
\end{algorithmic}
\end{algorithm}

In our experiments, we use two heuristic orderings of the vertices. One
is {\em minimum degree} (i.e., the next node is the node with the minimum
degree in the current graph). The other is {\em minimum fill in} (i.e.,
the next node is the node that requires the minimum number of added edges).
We try both heuristics and pick the tree decomposition with the smaller tree width.
It should be noted that our algorithm works with any tree decomposition.
It is just that in our experiments, we used the above two heuristics.

Figure~\ref{fig:tdcyclegraph} is an example attack graph, which will be used
as a running example for {\sc TDCycle}.

\begin{figure}[h!]
    \centering
  \includegraphics[width=0.5\linewidth]{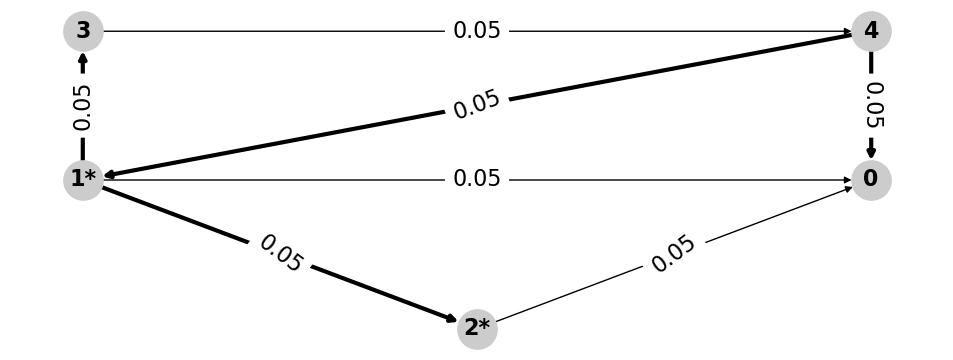}
    \caption{Example attack graph. Vertex $0$ is {\sc DA}. Vertex $1,2$ are entry vertices (marked using $*$).
    Edge labels represent the edges' failure rates.
    Thick edges (i.e., $1\rightarrow 2$) are not blockable.
    We set a budget of $2$.
    }
  \label{fig:tdcyclegraph}
\end{figure}

Figure~\ref{fig:tdcycleraw} is a tree decomposition of Figure~\ref{fig:tdcyclegraph}.

\begin{figure}[h!]
    \centering
  \includegraphics[width=0.5\linewidth]{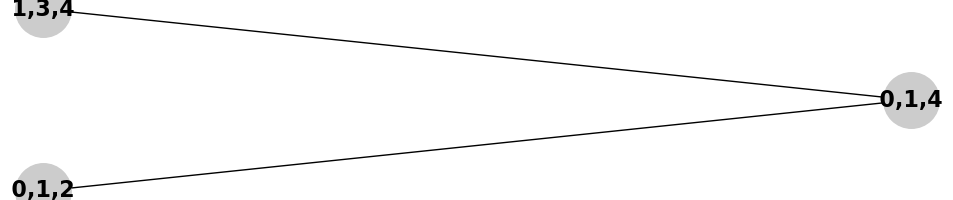}
    \caption{Tree decomposition of Figure~\ref{fig:tdcyclegraph}.
    Both {\em minimum degree} and {\em minimum fill in} result
    in the same tree decomposition.}
  \label{fig:tdcycleraw}
\end{figure}

We then convert the resulting tree decomposition into a {\em nice tree
decomposition}~\cite{Cygan2015:Parameterized}, as illustrated in
Figure~\ref{fig:tdcycletd}. The root node is a bag containing {\sc DA} only and
all the leaf nodes are bags of size one.\footnote{The standard nice tree decomposition definition involves
empty bags. For example, a leaf node in our illustration
contains one vertex, which can be interpreted as an introduce node
by adding an auxiliary empty bag as its child. We ignore empty bags
as they play no role in our DP.}
In a nice tree decomposition, there are only three types of nodes.

\begin{itemize}

    \item {\bf Introduce node}: An introduce node $X$ has an only child $X'$.
        $X'\subsetneq X$ and $X\setminus X'$ has size $1$.
        For example, in Figure~\ref{fig:tdcycletd}, $(3,4)$ is an introduce node.

    \item {\bf Forget node}: A forget node $X$ has an only child $X'$.
        $X\subsetneq X'$ and $X'\setminus X$ has size $1$.
        For example, in Figure~\ref{fig:tdcycletd}, $(1,4)$ is a forget node.

    \item {\bf Join node}: A join node $X$ has two children. $X$ must be identical
        to both children.
        For example, in Figure~\ref{fig:tdcycletd}, $(0,1,4)$ is a join node.

\end{itemize}

It is a trivial task to convert a tree decomposition into a nice tree
decomposition.  For example, in Figure~\ref{fig:tdcycleraw}, $(0,1,4)$ splits
into two branches.  We just need to add two clones of it to make it a valid
join node.  As another example, in Figure~\ref{fig:tdcycleraw}, $(1,3,4)$ is
connected to $(0,1,4)$ directly. By inserting $(1,4)$ in between, $(0,1,4)$
(top-left one among the three clones) becomes an introduce node and $(1,4)$
becomes a forget node.
A tree decomposition with tree width $w$
can be converted to a nice tree decomposition with $O(wn)$ nodes.

The whole point of nice tree decomposition is to divide the nodes into three
categories, so that when we design DP, we just need to come up with three rules
(one rule for each category of nodes).

\begin{figure}[h]
    \centering
  \includegraphics[width=0.5\linewidth]{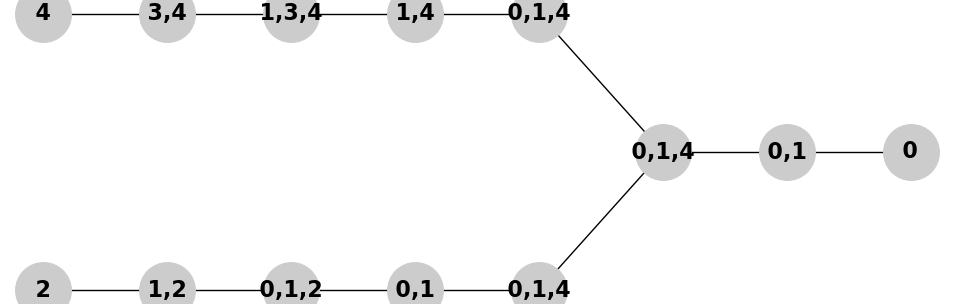}
    \caption{Nice tree decomposition built upon Figure~\ref{fig:tdcycleraw}.}
  \label{fig:tdcycletd}
\end{figure}

\subsection{{\sc TDCycle} example based on Figure~\ref{fig:tdcycletd}}

Before discussing {\sc TDCycle} in the context of Figure~\ref{fig:tdcycletd},
we first prove Lemma~\ref{lm:forget1to1}.

\begin{proof}[Proof of Lemma~\ref{lm:forget1to1}]
Let $(u,v)$ be an edge in the attack graph.
In any tree decomposition, there must be a bag $X$ that contains both $u$ and $v$.
We check the parent of $X$. If the parent also contains both $u$ and $v$,
then we keep going up, until one of $u$ or $v$ is forgotten. (Both cannot be
forgotten at the same node, since exactly one is forgotten at a forget node under
a nice tree decomposition.) The forget node where either $u$ or $v$ is forgotten
is the node we assign $(u,v)$ to.
The above forget node is unique for the following reasons.
A known property of tree decomposition is that a node is forgotten at a unique forget node.
The forget node for $u$ and $v$ must satisfy that either the forget node
    for $u$ is an ancestor of the forget node for $v$ (in which
    case $(u,v)$ is assigned to the forget node for $v$), or the reverse.
    This is because the subgraphs induced by $u$ or $v$ must be trees.
That is, at the above unique forget node, we decide whether or not to block $(u,v)$.
\end{proof}

Recall that in our DP, for every node $X$, there is a corresponding subproblem $DP(X)$,
which returns the collection of all possible tuples (distance matrix achieved
and budget spent) after processing $St(X)$, where $St(X)$ is the subtree rooted
at $X$. In our example, all edges have the same failure rate $0.05$, which
translates to a distance of $\delta=-\ln(1-0.05)$.

\begin{itemize}

    \item Base cases: There are only two leaf nodes $(4)$ and $(2)$.
        At leaf nodes, we have not put back any edges at all.
        $4$ is not an entry vertex, so after processing $(4)$, the
        only possible tuple is $([\infty],0)$.
        That is,
        \[DP((4))=\{([\infty],0)\}\]
        Similarly, since $2$ is an entry vertex, we have
        \[DP((2))=\{([0],0)\}\]

    \item Original problem: We have set a budget of $2$.
        The root node is $(0)$. After processing all nodes, we have
        \[DP((0))=\{([3\delta],2), ([\delta],0)\}\]
    That is, after processing all nodes in the nice tree decomposition,
        we end up with two possible tuples.
        We can spend $2$ units of budget
        to ensure that the attacker's distance to $0$ ({\sc DA}) is $3\delta$,
        which translates to a success rate of $(1-0.05)^3$.
        We can also spend $0$ units of budget
        to ensure that the attacker's distance to $0$ ({\sc DA}) is $\delta$,
        which translates to a success rate of $1-0.05$.
        Given a budget of $2$, we conclude that the attacker's worst success
        rate is $(1-0.05)^3$.

        In our implementation, we have a micro-optimisation step that filtered
        out $([\delta],1)$ as this tuple is worse off than $([\delta],0)$.
        Given two different tuples $(M_1,b_1)$ and  $(M_2,b_2)$, if $b_2\ge
        b_1$ and $M_1\ge M_2$ (every element of $M_1$ is larger or equal), then
        $(M_2,b_2)$ is discarded.

    \item Introduce node: $(3,4)$ is an introduce node. $3$ is introduced
        at this node, but its edges have not been put back yet, so $3$
        is disconnected. $3$ is also not an entry vertex.
        Since $DP((4))=\{([\infty],0)\}$, we have
        \[DP((3,4))=\{\left(
    \begin{bmatrix}
        \infty  & \infty \\
        \infty  & \infty \\
    \end{bmatrix}
        ,0\right)\}\]

    $(1,3,4)$ is another introduce node. $1$ is introduced
        at this node. Its edges have not been put back yet, but $1$ is
        an entry vertex.
        So we have
        \[DP((1, 3,4))=\{\left(
    \begin{bmatrix}
        0 & \infty  & \infty \\
        \infty & \infty  & \infty \\
        \infty & \infty  & \infty \\
    \end{bmatrix}
        ,0\right)\}\]
        The $0$ in the above matrix (top-left corner) indicates that the distance from $1$ to
        an entry vertex is $0$ (as $1$ is an entry vertex itself).

    \item Forget node: $(1,4)$ is a forget node. $3$ is forgotten here.
        That is, we need to consider putting back some edges.
        The two edges assigned to this forget node is $1\rightarrow 3$
        and $3\rightarrow 4$. We must put back $1\rightarrow 3$ as it is
        not blockable. For $3\rightarrow 4$, we either block it or not.
        If we do not block $3\rightarrow 4$, then we end up with the following
        possible tuple:
        \[\left(\begin{bmatrix}
        0 & 2\delta  \\
        \infty & 2\delta \\
    \end{bmatrix}
        ,0\right)\]

        The interpretation is that after putting back both $1\rightarrow 3$
        and $3\rightarrow 4$, the end result is as follows:
        \begin{itemize}

            \item top left $0$: it takes $0$ hops to go from an entry vertex to $1$

            \item top right $2\delta$: it takes $2$ hops to go from $1$ to $4$

            \item bottom left $\infty$: we cannot go from $4$ to $1$ (note: the edge $4\rightarrow 1$ has not been put back yet)

            \item bottom right $2\delta$: it takes $2$ hops to go from an entry
                vertex to $4$

        \end{itemize}

        If we block $3\rightarrow 4$, then we end up with the following
        possible tuple:
        \[\left(\begin{bmatrix}
        0 & \infty\\
        \infty & \infty\\
    \end{bmatrix}
        ,1\right)\]

        In summary, we have
        \[DP((1,4))=\left\{\left(\begin{bmatrix}
        0 & 2\delta  \\
        \infty & 2\delta \\
    \end{bmatrix}
        ,0\right),
        \left(\begin{bmatrix}
        0 & \infty\\
        \infty & \infty\\
    \end{bmatrix}
        ,1\right)\right\}\]

    \item Join node:
        In Figure~\ref{fig:tdcycletd}, there are three clones of
        $(0,1,4)$.
        We use $(0,1,4)^R$, $(0,1,4)^{T}$ and $(0,1,4)^{B}$
        to denote the one on the right, top and bottom, respectively.

        We directly present the values of $DP((0,1,4)^T)$ and $DP((0,1,4)^B)$.

$DP((0,1,4)^T)$ is
        \[
        \left\{\left(\begin{bmatrix}
            \infty & \infty & \infty \\
            \infty & 0 & 2\delta  \\
            \infty & \infty & 2\delta \\
    \end{bmatrix}
        ,0\right),
        \left(\begin{bmatrix}
            \infty & \infty & \infty \\
            \infty & 0 & \infty\\
            \infty & \infty & \infty\\
    \end{bmatrix}
        ,1\right)\right\}\]


$DP((0,1,4)^B)$ is
        \[\left\{\left(\begin{bmatrix}
            \delta & \infty & \infty \\
            2\delta & 0 & \infty\\
            \infty & \infty & \infty\\
    \end{bmatrix}
        ,0\right),
        \left(\begin{bmatrix}
            \infty & \infty & \infty \\
            \infty & 0 & \infty\\
            \infty & \infty & \infty\\
    \end{bmatrix}
        ,1\right)\right\}\]


        $DP((0,1,4)^R)$ is then the aggregation of the above.
        We take one tuple $(M^T,b^T)$ from $DP((0,1,4)^T)$
        and one tuple $(M^B,b^B)$ from $DP((0,1,4)^B)$.
        We add $(\min(M^T,M^B),b^T+b^B)$ into $DP((0,1,4)^R)$, provided
        that $b^T+b^B\le b$.
        $DP((0,1,4)^R)$ is

        \[\left\{\left(\begin{bmatrix}
            \delta & \infty & \infty \\
            2\delta & 0 & 2\delta \\
            \infty & \infty & 2\delta\\
    \end{bmatrix}
        ,0\right),
        \left(\begin{bmatrix}
            \delta & \infty & \infty \\
            2\delta & 0 & \infty\\
            \infty & \infty & \infty\\
    \end{bmatrix}
        ,1\right),\right.\]
        \[\left.\left(\begin{bmatrix}
            \infty & \infty & \infty \\
            \infty & 0 & 2\delta\\
            \infty & \infty & 2\delta \\
    \end{bmatrix}
        ,1\right),
        \left(\begin{bmatrix}
            \infty & \infty & \infty \\
            \infty & 0 & \infty\\
            \infty & \infty & \infty\\
    \end{bmatrix}
        ,2\right)\right\}\]

\end{itemize}

\subsection{Pseudocode of {\sc TDCycle}}

We present the pseudocode of {\sc TDCycle} in this section.

\begin{algorithm}[h!]
\caption*{{\sc TDCycle} base cases setup}
\textbf{Input}: Nice tree decomposition {\sc TD}\\
\begin{algorithmic}[1]
\FOR {leaf node $X=\{x\}$ in {\sc TD}}
    \IF {$x$ is an entry vertex}
        \STATE $DP(X)=\{([0],0)\}$
    \ELSE
        \STATE $DP(X)=\{([\infty],0)\}$
    \ENDIF
\ENDFOR
\end{algorithmic}
\end{algorithm}

\begin{algorithm}[h!]
\caption*{{\sc TDCycle} final step after finishing DP}
    \textbf{Input}: $DP((\textsc{DA}))$
\begin{algorithmic}[1]
    \FOR {$([d_{\textsc{DA,DA}}],b')$ in $DP((\textsc{DA}))$}
    \STATE Keep track of the maximum $d_{\textsc{DA,DA}}$
\ENDFOR
    \STATE Return $1-e^{-d_{\textsc{DA,DA}}}$
\end{algorithmic}
\end{algorithm}

\begin{algorithm}[h!]
\caption*{{\sc TDCycle} at introduce node}
\textbf{Input}: Nice tree decomposition {\sc TD}\\
Introduce node $X=(x_1,\ldots,x_k,y)$\\
$X$'s child $X'=(x_1,\ldots,x_k)$
\begin{algorithmic}[1]
\IF {$y$ is an entry vertex}
    \STATE $d_{yy}=0$
\ELSE
    \STATE $d_{yy}=\infty$
\ENDIF
\FOR {possible tuple in $DP(X')$}
    \STATE convert the tuple on the left (from $DP(X')$) to the tuple on the right
    (added into $DP(X)$):
\begin{footnotesize}
\[
    \left(\begin{bmatrix}
        d_{11} & \ldots & d_{1k}\\
        \ldots\\
        d_{k1} & \ldots & d_{kk}\\
    \end{bmatrix},b'\right)
    \rightarrow
    \left(\begin{bmatrix}
        d_{11} & \ldots & d_{1k} & \infty\\
        \ldots\\
        d_{k1} & \ldots & d_{kk} & \infty\\
        \infty  & \ldots & \infty & d_{yy} \\
    \end{bmatrix},b'\right)
\]
\end{footnotesize}
\ENDFOR
\end{algorithmic}
\end{algorithm}

\begin{algorithm}[h!]
\caption*{{\sc TDCycle} at forget node}
\textbf{Input}: Nice tree decomposition {\sc TD}\\
Forget node $X=(x_2,\ldots,x_k)$\\
$X$'s child $X'=(x_1,\ldots,x_k)$
\begin{algorithmic}[1]
\STATE generate all possible blocking policies for edges
    between $x_1$ and one of $x_2,\ldots,x_k$ (at most $2^{k-1}$ options)
    \FOR {blocking policy with budget $b''$}
    \FOR {possible tuple $(M,b')$ in $DP(X')$}
    \IF {$b'+b''>b$}
    \STATE continue
    \ENDIF
    \STATE convert the tuple on the left (from $DP(X')$) to the tuple on the right
    (added into $DP(X)$):
\begin{footnotesize}
\[
    \left(\begin{bmatrix}
        d_{11} & \ldots & d_{1k} \\
        \ldots\\
        d_{k1} & \ldots & d_{kk} \\
    \end{bmatrix},b'\right)
    \rightarrow
    \left(\begin{bmatrix}
        d_{22}' & \ldots & d_{2k}'\\
        \ldots\\
        d_{k2}' & \ldots & d_{kk}'\\
    \end{bmatrix},b'+b''\right)
\]
\end{footnotesize}

    \STATE call {\em all-pair shortest path} routine to calculate the $d_{ij}'$,
    which are updated distances considering the newly put back edges under
    the current blocking policy.
\ENDFOR
\ENDFOR
\end{algorithmic}
\end{algorithm}

\begin{algorithm}[h!]
\caption*{{\sc TDCycle} at join node}
\textbf{Input}: Nice tree decomposition {\sc TD}\\
Join node $X$\\
Two children $X_1$ and $X_2$
\begin{algorithmic}[1]
\FOR {possible tuple $(M_1,b_1)$ in $DP(X_1)$}
\FOR {possible tuple $(M_2,b_2)$ in $DP(X_2)$}
    \IF {$b_1+b_2\le b$}
    \STATE add $(M',b_1+b_2)$ into $DP(X)$, where $M'$
is the element-wise minimum between $M_1$ and $M_2$
    \ENDIF
\ENDFOR
\ENDFOR
\end{algorithmic}
\end{algorithm}

\subsection{Complexity of {\sc TDCycle}}

\begin{proof}[Proof of Theorem~\ref{thm:tdcycle}]
There are $O(wn)$ DP subproblems in total, which fall into three categories as follows.
    The return value of $DP(X)$ is a collection. We use a binary array to represent
    this collection. The size of the array is $H^{w^2}(b+1)$, or simply $O(H^{w^2}b)$.\footnote{In experiments, we used Python's builtin set instead, because most tuples are not possible.}

    \begin{itemize}

        \item Introduce node: We need to go through the child node's output array
            (at most $O(H^{w^2}b)$ tuples)
            and expand the distance matrix by one row and one column. The complexity
            is $O(w^2)$ for each tuple.
            There are at most $O(wn)$ introduce nodes. So the complexity for handling all
            introduce nodes is $O(H^{w^2}bw^3n)$.

        \item Forget node:
        We need to go through the child node's output array
            (at most $O(H^{w^2}b)$ tuples)
            and shrink the distance matrix by one row and one column. The complexity
            is $O(2^ww^3)$ for each tuple as we need to run an all pair shortest path
            for each blocking option and there are $O(2^w)$ options.
            There are at most $O(n)$ forget nodes (a vertex is forgotten once). So the complexity for handling all
            forget nodes is $O(2^wH^{w^2}bw^3n)$.

        \item Join node:
        We need to go through both children nodes' output arrays
            (at most $O(H^{2w^2}b^2)$ pairs of tuples).
            We need to perform an element-wise min operation on the distance
            matrices, which has a complexity of $O(w^2)$.
            There are at most $O(n)$ join nodes.
            So the complexity for handling all
            join nodes is $O(H^{2w^2}b^2w^2n)$.

We do not need to run an all-pair shortest path update at a join node.  If the
            join node has a bag size of $1$, then all-pair shortest path is not
            necessary.  If the bag size is more than $1$, then there will be
            forget nodes between the join node and root. We could wait until we
            reach a forget node to run the all-pair shortest path process.

    \end{itemize}

    Assuming $H$ is at least $2$, the overall complexity is then
            \[O(2^wH^{w^2}bw^3n+H^{2w^2}b^2w^2n)\]
            \[=O((2^ww+H^{w^2}b)H^{w^2}bw^2n)=O(H^{2w^2}b^2w^2n)\]

\end{proof}

It should be noted that the above is the worst case complexity.
In experiments, the set size of $DP(X)$ is often quite small (i.e.,
in the hundreds for {\sc R2000}).



\subsection{$-\ln(1-x)$, when restricted to a small interval, can be approximated using
two straight lines (one upper bound and one lower bound)}

\begin{figure}[h]
    \centering
  \includegraphics[width=0.5\linewidth]{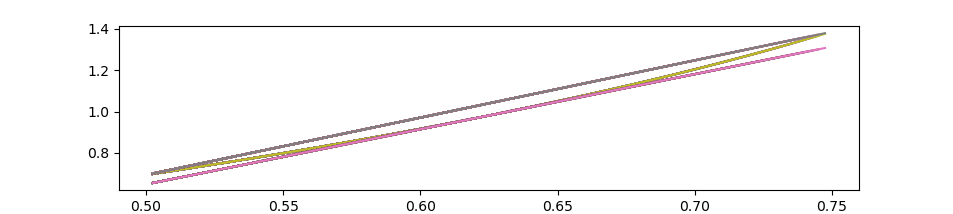}
    \caption{$-\ln(1-x)$ on $[0.5,0.75]$, bounded by straight lines}
  \label{fig:curve}
\end{figure}

\subsection{Reproducibility notes and preprocessing steps}

We have included our source code as well as the synthetic AD graphs.
{\sc R2000} and {\sc R4000} were generated using {\sc DBCreator} by setting
the numbers of computers to $2000$ and $4000$, respectively.
{\sc ADS10} was generated using {\sc adsimulator} with the following parameters:
nComputers=1000,
nUsers=1000,
nOUs=200,
nGroups=1000,
nGPOs=200.

Both tools are open source and are available at:

\noindent https://github.com/BloodHoundAD/BloodHound-Tools/tree/master/DBCreator

\noindent https://github.com/nicolas-carolo/adsimulator

In our experiments, we run $10$ trials with different
random setups (randomly picked blockable edges and entry nodes).  The
trials are using fixed random seeds from $0$ to $9$.
For {\sc MIP-Feasible} and {\sc MIP-LB}, we used $10$ regions to approximate $-\ln(1-x)$: $[0,0.1],[0.1,0.2],\ldots,[0.9,0.99]$.

Before running our algorithms, we perform the following preprocessing steps:

\begin{itemize}

    \item In both {\sc R2000} and {\sc R4000}, there are $7$ admin nodes.
        In {\sc ADS10}, there are $43$ admin nodes (nodes whose {\sc admincount}
        property is true).
        We
        merge them into a single destination node and set it as {\sc DA}.
        We
        ignore all out-going edges of {\sc DA}.  (Once the attacker reaches
        {\sc DA}, the attack has finished.)

    \item We delete all nodes that cannot reach {\sc DA}.

    \item We delete all incoming edges to entry nodes, because the optimal
        attack path will never involve such edges.

    \item We delete non-entry nodes that have $0$ in-degrees, because the
        attacker cannot start from these nodes or reach these nodes.

    \item If an edge is not block-worthy (Lemma~\ref{lm:bw}), then we
        mark it as not blockable.

    \item For a splitting node $u$, if there are two non-splitting paths
        $\textsc{nsp}(u,v_1)$ and $\textsc{nsp}(u,v_2)$ that satisfy: 1)
        these two paths have the same end node ($\textsc{dest}(u,v_1)=\textsc{dest}(u,v_2)$); 2) the first path has a higher success rate ($c_{u,v_1}\ge
        c_{u,v_2}$); 3) The first path $\textsc{nsp}(u,v_1)$ is not blockable, then the
        attacker always prefers the first path (facing any defense).
        We delete $u\rightarrow v_2$ as it is never used.

\end{itemize}

\subsection{Proximal Policy Optimization parameters and implementation details}

We used the PPO implementation from the popular {\sc Tianshou} library, with
the following parameters. Our code is also included.

\begin{footnotesize}
\begin{center}
    \begin{tabular}{ |l|l| }
 \hline
        learning rate & $0.001$ \\
 \hline
        discount factor & $0.99$ \\
 \hline
        batch size & $64$ \\
 \hline
        epoch & refer to experiment description \\
 \hline
        step per epoch & $10000$ \\
 \hline
        step per collect & $2000$ \\
 \hline
        repeat per collect & $10$ \\
 \hline
        hidden sizes & $[128, 128]$ for {\sc R2000} \\
                     & $[512, 512]$ for others \\
 \hline
\end{tabular}
\end{center}
\end{footnotesize}

The observation space is a tuple where every coordinate is between $0$ and $1$.
An observation consists of the following:

\begin{itemize}

    \item
Distance matrix before blocking: The length is the square of the maximum tree decomposition bag size. Every matrix element is expressed in terms of probability so it is between $0$ and $1$. If the current bag size is less than the maximum bag size, then we only
        use the ``top left'' corner. The remaining slots are filled in with $0$s (equivalent to adding auxiliary nodes that not connected to any other nodes).

    \item
Distance matrix after blocking.

    \item Budget spent: if the total budget is $5$ and $2$ is left, then
        we append $(1, 1, 0, 0, 0)$ to the observation.

    \item Current step index: if the total episode length is $5$ and we have
        finished $2$ steps, then we
        append $(1, 1, 0, 0, 0)$ to the observation.

\end{itemize}

\subsection{Experiments showing that our RL based approach is not merely performing
random searching and the tree decomposition based distance matrix indeed provides
a useful representation that facilitates learning}

We conduct experiments by running reinforcement learning on the original
observation space with the tree decomposition based distance matrix.  We then
repeat the experiments using the same setups and the same random seeds, but we
replace the actual distance matrix by the {\bf zero matrix} and by the {\bf
random matrix}.  For a fair comparison, we do not modify the representation for
the budget spent and the current step index.  Our experiments show that
learning significantly deteriorates after the replacements.

We design the following AD graph for our experiments.
Node $0$ is {\sc DA}. Besides {\sc DA}, we have $40$ entry nodes from node $1$ to node
$40$. Node $i$ ($1\le i\le 40$) has one {\em blockable} edge pointing to {\sc DA}.
These are all the edges in the graph.
We set a budget of $20$. That is, we need to block half of the edges.
For $i$ that is {\bf even}, the failure rate of edge $i\rightarrow \textsc{DA}$ is $\frac{40}{41}$, which represents an edge that does not need to be blocked as the failure rate is
high.
For $i$ that is {\bf odd}, the failure rate of edge $i\rightarrow \textsc{DA}$ is
$\frac{i}{41}$ ($\frac{1}{41}\le \frac{i}{41}\le \frac{39}{41}$), which represents an edge that needs to be blocked considering that our budget is $20$.
The design motivation of this AD graph is as follows:

\begin{itemize}

    \item
The task is to pick out $20$ ``correct'' edges to block out of $40$ edges, which
        has a large search space (over $10^{11}$).

    \item
The attacker's best success rate is $\frac{40}{41}$, which happens as long
        as $1\rightarrow \textsc{DA}$ is not blocked. That is, with this
        one mistake, we end up with the worst performance.
The attacker's worst success rate is $\frac{1}{41}$, which only happens
when every edge blocked is ``correct''.
        There is a wide range of possible results between $\frac{1}{41}$ and $\frac{40}{41}$ and it gets {\em increasingly difficult} to achieve better performance.

    \item We {\em intertwine} the block-worthy and not-block-worthy edges so
        that we do not accidentally learn toward policies such as ``spend all budget immediately''
        (if all block-worthy edges have lower indices) or ``spend all budget at
        the end'' (if all not-block-worthy edges have lower indices).
        It should be noted that the edges' indices are not included in the
        observation, so we will not accidentally learn toward ``block odd edges''.

\end{itemize}

We use {\sc Original} to represent the original approach with the
distance matrix intact.
We use {\sc Zero} and {\sc Random} to represent the results after replacing
the distance matrix using the zero matrix and the random matrix
(every matrix element independently drawn from $U(0,1)$).
We record the worst success rate for the attacker after $10$ and $50$ epochs.
Experiments are repeated $10$ times using training seed from $0$ to $9$.
The table cells are the averages and the standard deviations over these $10$ trials.
Clearly, the tree decomposition based distance matrix provides a useful
representation that facilitates learning (smaller numbers are better).

\begin{footnotesize}
\begin{center}
    \begin{tabular}{ |l|l|l|l| }
 \hline
       Epochs & {\sc Original} & {\sc Zero} & {\sc Random}\\
 \hline
        $10$ & $\mathbf{0.298\pm 0.071}$ & $0.492\pm 0.063$ & $0.478\pm 0.048$\\
 \hline
        $50$ & $\mathbf{0.117\pm 0.024}$ & $0.371\pm 0.039$ & $0.415\pm 0.040$\\
 \hline
\end{tabular}
\end{center}
\end{footnotesize}


Lastly, we run experiments on {\sc R2000} with budget $5$.
We use the same experimental setup as before (as described in the experiment
section). We randomly draw the
entry nodes and the blockable edges using random seed $0$ to $9$.
Experiments are repeated $10$ times using training seed from $0$ to $9$.
We record the worst success rate for the attacker after $5$ and $10$ epochs (due to time constraint -- as we are running three experiments for $10$ random graphs and for $10$ training seeds, which in total equals $300$ experiments).
For seed $1,4,6,7,8,9$, every trial reaches optimality within $5$ epochs, so
we only present the result for seed $0,2,3,5$.
The numbers do not differ much because 1) the gap between the optimal result and
the greedy result is tiny to begin with; 2) the search space isn't large (given enough time, random searching often can find the optimal result for this graph).
Still, we see that the original approach always produces the best result.

\begin{footnotesize}
\begin{center}
    \begin{tabular}{ |l|l|l|l| }
 \hline
        {\sc R2000}& {\sc Original} & {\sc Zero} & {\sc Random}\\
        $b=5$&&& \\
        $5$ epochs &&& \\
 \hline
        Seed $0$ & $\mathbf{0.521\pm 0}$ & $\mathbf{0.521\pm 0}$ & $\mathbf{0.521\pm 0}$\\
 \hline
        Seed $2$ & $\mathbf{0.521\pm 0}$ & $\mathbf{0.521\pm 0}$ & $\mathbf{0.521\pm 0}$\\
 \hline
        Seed $3$ & $\mathbf{0.406\pm 0.030}$ & $0.416\pm 0.040$ & $0.506\pm 0.021$\\
 \hline
        Seed $5$ & $\mathbf{0.392\pm 0.008}$ & $0.394\pm 0.006$ & $0.446\pm 0.050$\\
 \hline
\end{tabular}
\end{center}
\end{footnotesize}
\begin{footnotesize}
\begin{center}
    \begin{tabular}{ |l|l|l|l| }
 \hline
        {\sc R2000}& {\sc Original} & {\sc Zero} & {\sc Random}\\
        $b=5$&&& \\
        $10$ epochs &&& \\
 \hline
        Seed $0$ & $\mathbf{0.521\pm 0}$ & $\mathbf{0.521\pm 0}$ & $\mathbf{0.521\pm 0}$\\
 \hline
        Seed $2$ & $\mathbf{0.504\pm 0.043}$ & $0.521\pm 0$ & $0.521\pm 0$\\
 \hline
        Seed $3$ & $\mathbf{0.396\pm 0}$ & $\mathbf{0.396\pm 0}$ & $0.406\pm 0.030$\\
 \hline
        Seed $5$ & $\mathbf{0.388\pm 0.010}$ & $0.390\pm 0.009$ & $0.396\pm 0$\\
 \hline
\end{tabular}
\end{center}
\end{footnotesize}



\subsection{Description of the kernel}

Due to space constraint, we did not explicitly mention what the kernel is in our
section on kernelization. Also, for presentation purpose, we want to avoid
confusion by brining up a different graph.
For completeness, the kernel is as follows: Only keep splitting nodes, entry nodes and {\sc DA}.
Delete all original edges. Replace every original non-splitting path by one edge. The edge's success
rate is the unblocked success rate of the original non-splitting path. The edge is blockable
if and only if the original non-splitting is blockable.

\subsection{More on {\sc TDCycle} vs {\sc IP}}

As shown in our experiments, {\sc IP} scales extremely well and performs much
better than {\sc TDCycle} on all three graphs generated using
two different open source AD graph generators.
We expect {\sc TDCycle} to perform relatively better if we are dealing with a graph
with a small tree width and a large number of non-splitting paths.
As an example, let us consider the following contrived attack graph.
There is one entry node called {\sc entry} and one {\sc DA}. We insert the
following edges: {\sc entry} $\rightarrow i\rightarrow$ {\sc DA} for $i$ from
$1$ to $n-2$ (that is, there are in total $n$ nodes).
All edges are set to be blockable and have the same failure rates.
This attack graph has a tree width of $2$.
The maximum attack path length $l$ is $2$.
The number of possible success
rates for attack paths $H$ is $4$.
With the above set up, {\sc TDCycle}'s complexity becomes $O(b^2n)$.
On the contrary, for {\sc IP}, we face an integer program with $O(n)$ binary variables.

Lastly, {\sc TDCycle} has two advantages over {\sc IP}:

\begin{itemize}

    \item {\sc TDCycle} is {\em embarrassingly parallel}.
That is, if we throw $100$ CPU cores to it, we expect about $100$ times speed up.
        For all three types
        of nodes (introduce/forget/join), the main calculation is
        a for loop that can be trivially
        made parallel. Given a node $X$, we just need to use a shared memory array to store the result for $DP(X)$. All parallel processes can write to this shared array without worrying about race condition.

        On the contrary, {\sc IP} is difficult to be made parallel.

    \item {\sc TDCycle} does not require that we load the whole model into the
        memory.  Memory paging is convenient under {\sc TDCycle}.  The dynamic
        program is in a bottom-up fashion.  Node $X$ in the tree decomposition
        requires a memory block to store its result $DP(X)$.
        For node $X$, we
        only need to allocate memory for it once we reach it.  Also, after processing
        a node, we can safely delete the memory blocks for the children nodes.

        On the contrary, memory paging is difficult for {\sc IP}.

\end{itemize}


\end{document}